\documentclass{IEEEtran} %[onecolumn,draft,conference,twoside,12pt,onecolumn,draftclsnofoot(or "draft")]
\usepackage{tikz}
\usepackage{graphicx} 
\usepackage{amsmath,amssymb,amsthm,bm,mathtools}  %数学公式字体   %\usepackage{mathtools} provide {dcases}
\usepackage{soul} % underline multiple lines
\usepackage{tabularx}
\usepackage{arydshln} %%dashed lines in matrix
\usepackage{enumerate}
\usepackage{balance}   %align at bottom for two-column documents (also need a command \balance)
\usepackage[hidelinks]{hyperref}  %%generate catalog, automatically add parenthesis for equation citation  %option 'hidelinks' hides the citation links 
\usepackage{cleveref} % multiple refference
\usepackage{multirow}
\usepackage{subcaption}
\usepackage{cite}
\usepackage{url}
\allowdisplaybreaks[3]  %公式分页显示，数字0-4表示分开程度
\usepackage[normalem]{ulem} %%underline or delete
\usepackage{hhline}
\usepackage{diagbox}  %make table heads with diagonal lines
\DeclareGraphicsExtensions{.pdf,.png,.jpg,.eps}
\begin{document}
\title{On the Information Leakage in Private Information Retrieval Systems}
\author{Tao Guo, Ruida Zhou,~\IEEEmembership{Student Member,~IEEE}, and Chao Tian,~\IEEEmembership{Senior~Member,~IEEE}
%\author{\IEEEauthorblockN{Tao Guo, Ruida Zhou, and Chao Tian}
%	\thanks{The authors are with the Department of Electrical and Computer Engineering, Texas A\&M University, College Station, TX, USA. (\mbox{e-mail:} \href{mailto:guotao@tamu.edu}{guotao@tamu.edu}, \href{mailto:ruida@tamu.edu}{ruida@tamu.edu}, \href{mailto:chao.tian@tamu.edu}{chao.tian@tamu.edu})}
%	\thanks{The work of Ruida and C. Tian was supported in part by the National Science Foundation under Grant CCF-18-32309 and CCF-18-16546.}
}
\maketitle
%\balance
%%---------------------------------------newcommands-----------------------------------------------------------
\newcommand{\reffig}[1]{Figure \ref{#1}}
\newcommand{\cA}{\mathcal{A}}
\newcommand{\cK}{\mathcal{K}}
\newcommand{\cQ}{\mathcal{Q}}
\newcommand{\cS}{\mathcal{S}}
\newcommand{\cX}{\mathcal{X}}
\newcommand{\cY}{\mathcal{Y}}
\newcommand{\bE}{\mathbb{E}}
\newcommand{\bF}{\mathbf{F}}

\theoremstyle{plain}% default
\newtheorem{theorem}{Theorem}
\newtheorem{lemma}{Lemma}
\newtheorem{corollary}{Corollary}
\newtheorem{proposition}{Proposition}
\newtheorem{conjecture}{Conjecture}
\newtheorem{claim}{Claim}
\newtheorem{example}{Example}

\theoremstyle{remark}
\newtheorem{remark}{Remark}

\theoremstyle{definition}
\newtheorem{definition}{Definition}

\captionsetup[subfigure]{labelformat=simple, labelsep=space}
\renewcommand\thesubfigure{(\alph{subfigure})}

%%----------------------------------------- abstract -----------------------------------------------------------
% As a general rule, do not put math, special symbols or citations in the abstract
%%==========================================================================
\begin{abstract}
	We consider information leakage to the user in private information retrieval (PIR) systems. 
	Information leakage can be measured in terms of individual message leakage or total leakage. Individual message leakage, or simply individual leakage, is defined as the amount of 
	information that the user can obtain on any individual message that is not being requested,  %other than the required one. 
	and the total leakage is defined as the amount of information that the user can obtain about all the other messages except the one being requested. 	
	In this work, we characterize the tradeoff between the minimum download cost and the individual leakage, and that for the total leakage, respectively. 
	Coding schemes are proposed to achieve these optimal tradeoffs, which are also shown to be optimal in terms of the message size. 
	We further characterize the optimal tradeoff between the minimum amount of common randomness and the total leakage. 
	Moreover, we show that under individual leakage, common randomness is in fact unnecessary when there are more than two messages. 
\end{abstract}
%% Note that keywords are not normally used for peerreview papers. always use IEEEkeywords

%%=====================================================================
\section{Introduction}\label{section-intro}
The problem of {\it private information retrieval} (PIR) \cite{PIR-1,PIR-2} addresses the retrieval of one out of $K$ messages from $N$ replicated databases, 
without revealing the identity of the desired message to any individual database. 
A trivial solution is to simply retrieve all the messages, however, the download cost in this solution is unrealistic. Therefore, 
in the PIR problem, the goal is to find an efficient protocol, i.e., with the minimum download cost, to privately retrieve the desired message. 
The capacity of a PIR system is defined as the maximum number of bits of desired message that can be retrieved per bit of downloaded information, 
which was shown in \cite{Sun-Jafar-PIR-capacity-2017IT} to be $C_{\text{PIR}}=\left(1+1/N+1/N^2+\cdots+1/N^{K-1}\right)^{-1}$. %The optimal code  constructed in \cite{Sun-Jafar-PIR-capacity-2017IT} will be referred to as the SJ code in the rest of the paper. 
%A class of capacity-achieving codes proposed in \cite{Tian-Sun-Chen-PIR-IT19}, referre to as the TSC code, has the additional advantage of being also optimal in both the message size and the upload cost (query transmission cost). 
There were also significant interests recently on various other variations of the PIR problem; see, e.g., 
\cite{Tian-Sun-Chen-PIR-IT19,Sun-Jafar-SPIR-19IT,Sun-Jafar-robustPIR-2017IT,Sun-Jafar-PIR-2017TIFS,Sun-Jafar-PIR-multiround-2018IT,Sun-Jafar-conjectured-2018IT,Banawan-Ulukus-codedPIR-2018IT,Banawan-Ulukus-PIR-MultiMessage-18IT,Banawan-Ulukus-PIR-byzatine-colluding-IT19,Wangqiwen-SPIR-MDS2016ICC,Wangqiwen-SPIR-MDS-2019IT,Wangqiwen-PIR-eva-2019IT,Wangqiwen-PIR-SPIR-Eva-2019IT,secure-and-PIR-2018TVT,coded-PIR--TChan-2015ISIT,distributedPIR,PIR-MDS-dis-storage-2018IT,PIR-array-code-2019IT,PIR-asymmetry-IT19,PIR-cache-Allerton17,PIR-coded-colluding-17,PIR-dis-storage-18TIFS,PIR-dis-storage-linear-ISIT17,PIR-dis-storage-MDS-19IT,Zhouruida-PIR-arxiv19}.

The problem considered in this work is closely related to the {\it symmetric} private information retrieval (SPIR)  problem~\cite{Sun-Jafar-SPIR-19IT}, 
where ``symmetric" refers to the fact that both user privacy and database privacy need to be preserved. 
Database privacy requires that the user obtains no information on other messages beyond the requested message; strictly speaking, this is a security requirement rather than a privacy requirement, and we shall refer to it as such in the sequel.  
It is well understood that SPIR can also be viewed as a multi-database oblivious transfer problem~\cite{oblivious-transfer1,oblivious-transfer2,ishai2008founding,naor2000distributed}. 
It was shown in \cite{Sun-Jafar-SPIR-19IT} that to ensure perfect security, 
the databases need to share {\it common randomness} (a common key) which is independent of the messages and only available to the databases. 
The common key is used for randomizing the answers such that no information about the non-desired messages is leaked. 
The capacity of SPIR was shown to be $C_{\text{SPIR}}=1-1/N=\left(1+1/N+1/N^2+\cdots+1/N^{\infty}\right)^{-1}$, 
as long as the amount of common randomness is at least $\frac{1}{N-1}$ bits per desired message bit. 

In practice, the requirement of perfect security may be too stringent. For example, in the SPIR setting, it is likely to be acceptable for a user to obtain a single byte of information in a 1-hour video data (assuming the message is a video) that he does not request, i.e., leaking a very small amount of information \cite{yamamoto86,raymondbook,guo-guang-shum-2018ITW}. 
This is the type of systems that we wish to understand in this work. Allowing a small amount of information leakage let us control the system security level in a finer grain manner. In this context, SPIR essentially requires strictly zero information leakage, while the classical PIR does not have a security constraint at all. Thus our goal is to understand the tradeoff between the download cost and the amount of information leakage in the regime other than these two extreme cases. 
From the perspective of the oblivious transfer, we are essentially considering a generalized multi-database oblivious transfer setting \cite{naor2000distributed}, where one side of the privacy requirements (that corresponding to the database privacy) is allowed to be less than perfect. %Note that the information leakage is from the information security point of view in the distributed storage system and cannot be viewed as generalizations of the oblivious transfer model, which requires perfect security in the single-database case.}

Two different notions of information leakage can be defined: the individual leakage is defined as the amount of information that the user  can obtain about any individual non-desired message, and the total leakage as the amount of information that the user can obtain about all the non-desired messages. The former is similar to the weak security constraints in the literature \cite{weakly-NC-Narayanan05,weakly-NC-Alex11,weakly-SMDC-GTLY-19,weakly-secure-channel1}, while the latter corresponds to the standard strong security constraints. 
In practice, the weak security requirement may be more suitable than the strong version in some applications. For example, when the messages are video sequences, the user should not obtain information about any individual video segment, but obtaining the binary XOR of two video sequences may not be an issue since it will not lead to a meaningfully decodable video sequence, and thus enforcing that the user cannot obtain this XOR may be too stringent. Moreover, a protocol with a weak security constraint can potentially be implemented more efficiently in practical settings and may not require common keys.

The main result of this work is the characterizations of the optimal tradeoffs between the download cost and the amount of information leakage for both individual leakage and total leakage. By adapting the codes given in \cite{Tian-Sun-Chen-PIR-IT19} and \cite{Sun-Jafar-SPIR-19IT}, we provide code constructions that can achieve these optimal tradeoffs; moreover, they are also shown to have the minimum message sizes. For the individual leakage case, the constructed codes do not require common randomness (unless there are only two messages). At the extreme case with perfect individual message security, the download cost is in fact the same as that with perfect total security. This is rather reassuring since it implies that the stronger security requirement does not induce any additional download cost.

A related problem, namely weakly private information retrieval, was studied in \cite{weakly-PIR-isit2019,leaky-PIR-isit2019,WeakPIR-jiazhuqing-thesis}. For weakly private information retrieval systems, a small amount of privacy leakage is allowed, i.e., the leakage is to the server and on the requested message index, and it differs from the setting we are considering where the information leakage is to the user and on the message content. 

The rest of the paper is organized as follows. We formally define the problem and review the performance of classical PIR and SPIR codes in \Cref{section-preliminary}.
\Cref{section-main} is devoted to our main results on the optimal tradeoffs and the minimum message size.
%In \Cref{section-WS-PIR-code}, we propose an optimal code for WS-PIR that do not need any common randomness for $K\geq 3$. 
The proofs are given in \Cref{section-strong-proof}, \Cref{section-weak-proof}, and the appendices.
%In \Cref{section-discuss}, we discuss the PIR system with worst-case information leakage. 
We conclude the paper in \Cref{section-conclusion}.

%%====================================================================
\section{Preliminaries}\label{section-preliminary}
\subsection{Problem Statement}\label{section-setup}
For positive integers $K,N$, let $[1:K]\triangleq \{1,2,\cdots,K\}$ and $[1:N]\triangleq\{1,2,\cdots,N\}$. 
For any $k\in[1:K]$, define the complement set by $\bar{k}\triangleq \{1,2,\cdots,K\}\backslash\{k\}$. 

In a {\it private information retrieval} (PIR) system, there are $K$ independent messages $W_{1:K}=(W_1,W_2,\cdots,W_K)$, 
each of which is comprised of $L$ i.i.d. symbols uniformly distributed over a finite alphabet $\cX$. 
In the sequel, we will use $\log_{|\cX|}$-ary units, instead of bits, as the information measure, unless specified otherwise. We have 
\begin{align}
&H(W_{1:K})=H(W_1)+H(W_2)+\cdots+H(W_K), \\
&H(W_1)=H(W_2)=\cdots=H(W_K)=L.
\end{align} 
There are a total of $N$ databases, each of which stores all the messages $W_{1:K}$. 
A user aims to retrieve a message $W_k,~k\in[1:K]$ from the $N$ databases 
without revealing the identity $k$ of the desired message to any individual database. 
An independent random key $\bF$ is used to generate queries $Q_{1:N}^{[k]}=\left(Q_1^{[k]},Q_2^{[k]},\cdots,Q_N^{[k]}\right)$, i.e.,
\begin{equation}
H(Q_{1:N}^{[k]}|\bF)=0,~\forall k\in[1:K],  \label{query-generation}
\end{equation} 
where $Q_n^{[k]}\in\cQ_n$ for $n\in[1:N]$. 
For $n\in[1:N]$, the $n$-th query $Q_n^{[k]}$ is sent to the $n$-th database. 
Since we wish to protect the non-desired messages from information leakage, the databases may need to share a common random key $S\in\cS$ that is not accessible to the user, which induces the condition
\begin{equation}
H(W_{1:K},\bF,S)=H(W_{1:K})+H(\bF)+H(S).
\end{equation}
Upon receiving $Q_n^{[k]}$, the $n$-th database generates an answer $A_n^{[k]}$ from the query $Q_n^{[k]}$, the stored messages $W_{1:K}$, and the common key $S$, i.e., 
\begin{equation}
A_n^{[k]}=\varphi_n(Q_n^{[k]},W_{1:K},S),~\forall k\in[1:K],~\forall n\in[1:N], 
\end{equation} 
which implies 
\begin{align}
H(A_n^{[k]}|Q_n^{[k]},W_{1:K},S)=0,~\forall k\in[1:K],\forall n\in[1:N].   \label{ans-generation}
\end{align} 
The answer symbols are from a finite alphabet $\cY$, i.e., $A_n^{[k]}\in\cY^{\ell_n}$, where $\ell_n$ is the length of the answer. 
Using all the answers $A_{1:N}^{[k]}=\left(A_1^{[k]},A_2^{[k]},\cdots,A_N^{[k]}\right)$ from the $N$ databases and the values of $\bF$ and $k$, the user perfectly decodes the desired message $W_k$, which further implies that 
\begin{equation}
H(W_k|A_{1:N}^{[k]},\bF)=0.  \label{recovery-constraint}
\end{equation}

To satisfy the privacy requirement of keeping the desired message index private to any one of the databases, 
the received queries should be identically distributed, i.e., 
\begin{equation}
Q_n^{[k]}\sim Q_n^{[k']},~\forall k,k'\in[1:K],~\forall n\in[1:N].  
\end{equation}
Since $W_{1:K}$, $S$, and $\bF$ are independent, in light of \eqref{query-generation}, we have 
\begin{align}
&(Q_n^{[k]},W_{1:K},S)\sim(Q_n^{[k']},W_{1:K},S),   \nonumber \\
&\qquad \qquad \qquad \forall k,k'\in[1:K],~\forall n\in[1:N]. 
\end{align}
Since $A_n^{[k]}$ is a deterministic function of $(Q_n^{[k]},W_{1:K},S)$, the following identical distribution constraint must also hold 
\begin{align}
&(Q_n^{[k]},A_n^{[k]},W_{1:K},S)\sim(Q_n^{[k']},A_n^{[k']},W_{1:K},S),  \nonumber \\
&\qquad \qquad \qquad \forall k,k'\in[1:K],~\forall n\in[1:N].  \label{privacy-identical-distribute}
\end{align}

In contrast to the perfect security of the non-desired messages $W_{\bar{k}}=(W_1,\cdots,W_{k-1},W_{k+1},\cdots,W_K)$ in \cite{Sun-Jafar-SPIR-19IT}, 
we allow information leakage in this work. Define the total leakage as $I(W_{\bar{k}};Q_{1:N}^{[k]},A_{1:N}^{[k]},\bF)$, 
which is the amount of information that the user can obtain about all the non-desired messages. 
Define the individual message leakage for message $k'\in[1:K]~(k'\neq k)$ as $I(W_{k'};Q_{1:N}^{[k]},A_{1:N}^{[k]},\bF)$, 
which is the amount of information that the user can obtain about  the individual non-desired message-$k'$. 
The total leakage and the individual leakage in the systems are constrained as follows
\begin{align}
\frac{1}{L}I(W_{\bar{k}};Q_{1:N}^{[k]},A_{1:N}^{[k]},\bF)&\leq s,~ \forall k\in[1:K]    \label{def-strong}  \\
\frac{1}{L}I(W_{k'};Q_{1:N}^{[k]},A_{1:N}^{[k]},\bF)&\leq w,~ \forall k'\neq k\in[1:K],   \label{def-weak}
\end{align}
where the parameters $s$ and $w$ are used to indicate the strong security requirement and the weak security requirement, respectively. 
For $K=2$, since $I(W_{\bar{k}};Q_{1:N}^{[k]},A_{1:N}^{[k]},\bF)=I(W_{k'};Q_{1:N}^{[k]},A_{1:N}^{[k]},\bF)$, the total leakage constraint is equivalent to the individual leakage constraint. 

%\begin{remark}
%	$\frac{1}{(K-1)L}I(W_{\bar{k}};A_{1:N}^{[k]},Q_{1:N}^{[k]})$ is the normalized information leakage per message. 
%	It may also be of interest in some other circumstances. 
%\end{remark}

In a PIR system, the message size is defined as the number of bits to represent each individual message, which is clearly $L\log_2|\cX|$ in our notation. 
The download cost is defined as 
\begin{equation}
D\triangleq \log_{|\cX|}|\cY|\sum_{n=1}^N\bE(\ell_n),  \label{def-D}
\end{equation}
where the expectation is taken over the set $\cQ_n$ of all possible queries. 
Note that $D$ is a deterministic function of queries and query distribution, but neither the particular realization of messages nor the choice of desired message index $k$. 
%Our main results here is to characterize the tradeoff between the minimum download cost $D_{\min}$ and the security levels $s$ and $w$. 
%The retrieval rate of the system characterizes the amount of retrieved information bit per downloaded bit, and is defined by 
%\begin{equation}
%R\triangleq \frac{L}{D}=\frac{L}{\log_{|\cX|}|\cY|\sum_{n=1}^N\bE(\ell_n)}. 
%\end{equation}
The amount of common randomness is normalized by the message length $L$ as
\begin{equation}
\rho\triangleq \frac{H(S)}{L}.
\end{equation}

%%=====================================================================
\subsection{Analysis of Several Existing PIR Codes}\label{section-existing}
Before presenting our main result, we provide the result of a simple analysis on the capacity-achieving PIR code given by Sun and Jafar \cite{Sun-Jafar-PIR-capacity-2017IT}, the capacity-achieving code proposed by Tian et al. \cite{Tian-Sun-Chen-PIR-IT19}, and the SPIR code \cite{Sun-Jafar-SPIR-19IT}; these three classes of codes will be referred to as the SJ code, the TSC code, and the SPIR code, respectively.
%In the following, we first consider total leakage normalized by message length $L$. 
\begin{itemize}
	\item The SJ code \cite{Sun-Jafar-PIR-capacity-2017IT}: 
	\begin{align}
	&L=N^K;  ~ D=\frac{N(N^K-1)}{N-1};   \label{SJ-performance}  \\
	&\frac{1}{L}I(W_{\bar{k}};A_{1:N}^{[k]},\bF)=\frac{1}{N-1} \left(1-\frac{1}{N^{K-1}}\right);  \label{SJ-total-leakage}  \\
	&\frac{1}{L}I(W_{k'};A_{1:N}^{[k]},\bF)=\frac{1}{N^{K-1}},~ \forall k'\neq k\in[1:K]. \label{SJ-individual-leakage}
	\end{align}
	
	\item The TSC code \cite{Tian-Sun-Chen-PIR-IT19}: 
	\begin{align}
	&L=N-1;  ~ D=\frac{N^K-1}{N^{K-1}};   \label{TSC-performance}  \\
	&\frac{1}{L}I(W_{\bar{k}};A_{1:N}^{[k]},\bF)=\frac{1}{N-1}\left(1-\frac{1}{N^{K-1}}\right);  \label{TSC-total-leakage}  \\
	&\frac{1}{L}I(W_{k'};A_{1:N}^{[k]},\bF)=\frac{1}{N^{K-1}},~ \forall k'\neq k\in[1:K].  \label{TSC-individual-leakage}
	\end{align}
	
	\item The SPIR code \cite{Sun-Jafar-SPIR-19IT}: 
	\begin{align}
	L=N-1; ~ D=N; ~ \rho=\frac{1}{N-1}.   \label{SPIR-performance}
	\end{align}
\end{itemize}

From \eqref{SJ-total-leakage}, \eqref{TSC-total-leakage} and \eqref{SJ-individual-leakage}, \eqref{TSC-individual-leakage}, it is seen that the SJ code and TSC code have the same performance in terms of both the total leakage and the individual leakage. 
As an example, consider the case $N=3,~K=3$, then $\frac{1}{L}I(W_{\bar{k}};A_{1:N}^{[k]},Q_{1:N}^{[k]},\bF)=\frac{4}{9}$ and $\frac{1}{L}I(W_{k'};A_{1:N}^{[k]},Q_{1:N}^{[k]},\bF)=\frac{1}{9}$ for both the SJ code and the TSC code. 

%The individual leakage normalized by message length is given below, where we see that the two codes perform the same as each other. 
%\begin{itemize}
%	\item SJ code \cite{Sun-Jafar-PIR-capacity-2017IT}: 
%	\begin{equation}
%	\frac{1}{L}I(W_{k'};Q_{1:N}^{[k]},A_{1:N}^{[k]})=\frac{N}{N^K}=\frac{1}{N^{K-1}},~ \forall k'\neq k\in[1:K].
%	\end{equation}
%	\item TSC code \cite{Tian-Sun-Chen-PIR-IT19}: 
%	\begin{equation}
%	\frac{1}{L}I(W_{k'};Q_{1:N}^{[k]},A_{1:N}^{[k]})=\frac{1}{N-1}\left[\left(\frac{1}{N}\right)^{K-2} \left(\frac{N-1}{N}\right)\right]=\frac{1}{N^{K-1}},~ \forall k'\neq k\in[1:K]. 
%	\end{equation}
%\end{itemize}
%
%The download costs of the two capacity-achieving codes are given in as follows.

%%======================================================================
\section{Main Results}\label{section-main}
%For $K=1$, there is only one message, and the privacy constraint is impossible since the databases always know which message is desired. 
%For $N=1$, there is only one database, we are not able to retrieve a message by keeping privacy of its index and security of other messages simultaneously. 
%This can be seen as follows: i) the privacy constraint in \eqref{privacy-identical-distribute} implies that the answer from the only database $A_1^{[k]}$ is identically distributed for all $k\in[1:K]$; ii) the recovery constraint in \eqref{recovery-constraint} implies that $W_k$ can be decoded from $A_1^{[k]}$ and $\bF$, and thus we can decode all messages $W_{1:K}$ from $A_1^{[k]}$ and $\bF$, which means that the non-desired messages are not secure any more. 
%
%In the sequel, we consider only $K,N\geq 2$. 
%%The tradeoff between the minimum download cost $D_{\min}$ and total leakage level $s$ and individual leakage level $w$ are characterized, respectively.

%%===============================================================

%We present the result on the total leakage and the individual leakage next. 

\subsection{Total Leakage}\label{section-total}
%When measuring the security of the PIR system on total leakage, 
%an upper bound on the normalized leakage called total leakage level $s$ is considered. 

The following theorem characterizes the optimal tradeoff between the minimum download cost $D_{\min}$ and total leakage constraint $s$. 
For notational convenience, define 
\begin{equation}
D_{\min}^0\triangleq L\cdot\left(1+\frac{1}{N}+\cdots+\frac{1}{N^{K-1}}\right). \label{def-D-min-0}
\end{equation}
%%%--------------------- BEGIN thm-tradeoff-strong ------------------------------
\begin{theorem}\label{thm-tradeoff-strong}
	If the amount of common randomness satisfies $\rho\geq \rho_{\min}^{\text{s}}$, where 
	\begin{equation}
	\rho_{\min}^{\text{s}}\triangleq \frac{1}{N-1}-\frac{N^{K-1}}{N^{K-1}-1}\cdot s,  \label{min-key-size-strong}
	\end{equation}
	then the minimum download cost $D_{\min}$ of the PIR system is given by 
	\begin{equation}
	D_{\min}=
	\begin{cases}
	L\cdot\left(\frac{N}{N-1}-\frac{1}{N^{K-1}-1}\cdot s\right),&\text{if }0\leq s\leq s_t \\
	D_{\min}^0,&\text{otherwise,}
	\end{cases}  \label{D-min-strong}
	\end{equation}
	where the threshold is defined by $s_t\triangleq \frac{1}{N-1}\left(1-\frac{1}{N^{K-1}}\right)$  %$=\frac{1}{N}+\frac{1}{N^2}+\cdots+\frac{1}{N^{K-1}}$ 
	and $D_{\min}^0$ is defined in \eqref{def-D-min-0}. 
	If $\rho<\rho_{\min}^{\text{s}}$, then $D_{\min}=\infty$. 
\end{theorem} 
%%%--------------------- END thm-tradeoff-strong ------------------------------
\begin{proof}
	The proof is given in \Cref{section-strong-proof}. 
\end{proof}
\begin{remark}\label{remark-strong-thm}
	The case where $D_{\min}=\infty$ indicates that it is impossible to simultaneously meet all the system requirements, i.e., i) retrieval; ii) privacy; iii) total leakage constraint. 
	In this case, the capacity $C=0$.
\end{remark}

In light of \eqref{min-key-size-strong} and \eqref{D-min-strong}, we can view $D_{\min}$ as a bivariate function of $\rho$ and $s$. 
For a given $s$, the dependency of $D_{\min}$ on $\rho$ only appears at the threshold $\rho_{\min}^{\text{s}}$ 
so that i) $D_{\min}=\infty$ if $\rho<\rho_{\min}^{\text{s}}$; ii) $D_{\min}<\infty$ and $D_{\min}$ is independent of $\rho$ if $\rho\geq \rho_{\min}^{\text{s}}$. 
For $\rho\geq \rho_{\min}^{\text{s}}$, the dependency of $D_{\min}$ on $(\rho,s)$ is illustrated in Fig.~\ref{fig-tradeoff-Dmin-rho-s}. 
%%------------------------ BEGIN-figure-tradeoff-strong ---------------------------------
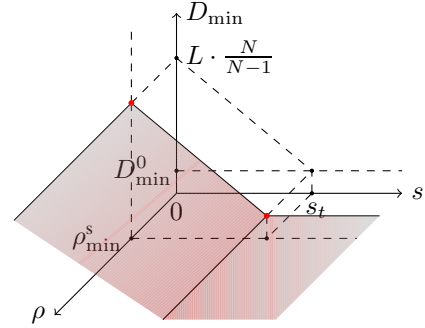
\begin{figure}[!ht]
	\centering
	\begin{tikzpicture}[scale=0.6]
	\draw[->] (0,0)--(0,4); \node [right] at (0,4) {$D_{\min}$};
	\draw[->] (0,0)--(5,0); \node [right] at (5,0) {$s$};
	\draw[->] (0,0)--(-2.7,-2.7); \node [left] at (-2.7,-2.7) {$\rho$};
	\node [below] at (0,0) {$0$};
	
	\draw (-1,2)--(-3.6,-0.6);  \draw[dashed] (-1,2)--(-1,3.6); 
	\draw (2,-0.5)--(-1,2);
	\draw (4.5,-0.5)--(2,-0.5)--(-0.3,-2.8);
	
	\draw[dashed] (-1,2)--(0,3)--(3,0.5)--(3,0); \filldraw (0,3) circle (1pt); \node [right] at (0,3) {$L\cdot\frac{N}{N-1}$}; 
	\draw[dashed] (0,0.5)--(3,0.5)--(5,0.5); \filldraw (3,0.5) circle (1pt); 
	\filldraw (0,0.5) circle (1pt); \node [left] at (0.1,0.5) {$D_{\min}^0$}; 
	\draw[dashed] (2,-0.5)--(3,0.5); \filldraw (3,0) circle (1pt); \node [below] at (3.1,0) {$s_t$}; 
	\draw[dashed] (-1,2)--(-1,-1); \filldraw (-1,-1) circle (1pt); \node [left] at (-1,-1) {$\rho_{\min}^{\text{s}}$}; 
	\draw[dashed] (-1,-1)--(2,-1)--(3,0) (2,-0.5)--(2,-1)--(4,-1); \filldraw (2,-1) circle (1pt); 
	
	\filldraw[red] (-1,2) circle (1.5pt); \filldraw[red] (2,-0.5) circle (1.5pt); 
	%	\fill[gray!20,nearly transparent] (-3.6,-0.6)--(-1,2)--(2,-0.5)--(-0.7,-3.2)--cycle; 
	%	\fill[gray!20,nearly transparent] (2,-0.5)--(-0.3,-2.8)--(1.8,-3.2)--(4.5,-0.5)--cycle; 
	\fill[top color=gray!10, bottom color=red,shading=axis,shading angle=45,opacity=0.15] (-3.6,-0.6)--(-1,2)--(2,-0.5)--(-0.3,-2.8)--cycle; 
	\fill[top color=red, bottom color=gray!10,shading=axis,shading angle=90,opacity=0.15] (2,-0.5)--(-0.3,-2.8)--(2.2,-2.8)--(4.5,-0.5)--cycle; 
	%	\node [right] at (2.5,1.3) {$D_{\min}^0\triangleq L\cdot\left(1+\frac{1}{N}+\cdots+\frac{1}{N^{K-1}}\right)$}; 
	\end{tikzpicture}
	\caption{The dependency of $D_{\min}$ on $(\rho,s)$.}
	\label{fig-tradeoff-Dmin-rho-s}
\end{figure}
%%-------------------------- END-figure-tradeoff-strong ---------------------------------
The dashed lines project the corner points to the axises that show their values. 
The value of $D_{\min}$ with respective to $(\rho,s)$ is given by the shaded area. 
We see that for a given value of $s$, $D_{\min}$ is a constant and thus independent with $\rho$. 
%For a given value of $\rho$ such that  $\rho \geq \rho_{\min}^{\text{s}}$, t
The shaded area is projected onto the $D_{\min}$-$s$ plane, which is drawn in Fig.~\ref{fig-tradeoff-strong}. 

%%------------------------ BEGIN-figure-tradeoff-strong ---------------------------------
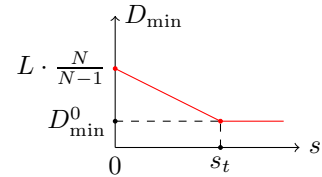
\begin{figure}[!ht]
	\centering
	\begin{tikzpicture}[scale=0.7]
	\draw[->] (0,0)--(0,2.5); \node [right] at (0,2.5) {$D_{\min}$};
	\draw[->] (0,0)--(3.5,0); \node [right] at (3.5,0) {$s$};
	\node [below] at (0,0) {$0$};
	
	\filldraw [red](0,1.5) circle (1pt); \node [left] at (0,1.5) {$L\cdot\frac{N}{N-1}$};
	\filldraw (2.0,0) circle (1pt); \node [below] at (2.0,0) {$s_t$}; \draw[dashed] (2.0,0)--(2.0,0.5);
	\filldraw (0,0.5) circle (1pt); \node [left] at (0,0.5) {$D_{\min}^0$}; \draw[dashed] (0,0.5)--(2.0,0.5);
	
	\draw[red] (0,1.5)--(2.0,0.5)--(3.2,0.5); \filldraw[red] (2.0,0.5) circle (1pt);
	
	%	\fill[green!20] (0,2.5)--(0,1.5)--(2.0,0.5)--(3.5,0.5)--cycle;
	\end{tikzpicture}
	\caption{Tradeoff curve between $D_{\min}$ and $s$.}
	\label{fig-tradeoff-strong}
\end{figure}
%%-------------------------- END-figure-tradeoff-strong ---------------------------------
The red line in Fig.~\ref{fig-tradeoff-strong} is the tradeoff curve  between $D_{\min}$ and $s$ for $\rho\geq \rho_{\min}^{\text{s}}$. Codes that achieve points on this optimal tradeoff curve will be referred to as {\it Pareto optimal} codes.
We see from Fig.~\ref{fig-tradeoff-strong} that for $s=0$, the problem reduces to SPIR \cite{Sun-Jafar-SPIR-19IT} 
for which the capacity is $C_{\text{SPIR}}=\left.\frac{L}{D_{\min}}\right|_{s=0}=1-\frac{1}{N}$. 
We further observe that i) the threshold $s_t$ is equal to the normalized total leakage of 
the capacity-achieving SJ code and TSC code in \eqref{SJ-total-leakage} and \eqref{TSC-total-leakage}; 
ii) for $s\geq s_t$, the capacity $C=\left.\frac{L}{D_{\min}}\right|_{s=s_t}=1+\frac{1}{N}+\cdots+\frac{1}{N^{K-1}}$ is equal to $C_{\text{PIR}}$. 
%Since $C$ is non-increasing in $s$ and lower bounded by $C_{\text{PIR}}$, 
%the two observations i) and ii) imply that for $s>s_t$, the total leakage level constraint is inactive, 
%and the Pareto optimal codes for $s$ can always achieve a security level of $s_t$ with either the SJ code or the TSC code. 

%It is important to point out that the curve between the two extreme points is in fact a straight line, 
%which means that any $(D_{\min},s)$ pair on this line is a linear combination of the two extreme points. 
%This observation provides a simple Pareto optimal code design, 
%which is a time sharing between the optimal classical PIR code and SPIR code. 
%The details of the code are given in \Cref{section-strong-achievability}. 

%%======================================================================
\subsection{Individual Leakage}\label{section-individual}
%We consider an upper bound on the normalized individual leakage called individual leakage level $w$. 
The following theorem characterizes the tradeoff between the minimum download cost $D_{\min}$ and the individual leakage constraint $w$. 
%%%--------------------- BEGIN thm-tradeoff-weak ------------------------------
\begin{theorem}\label{thm-tradeoff-weak}
	If the amount of common randomness satisfies $\rho\geq \rho_{\min}^{\text{w}}$, where 
	\begin{equation}
	\rho_{\min}^{\text{w}}\triangleq  
	\begin{cases}
	\frac{1}{N-1}-\frac{N}{N-1}\cdot w,& \text{ if }K=2 \\
	0, &\text{ if }K\geq 3,
	\end{cases}  \label{min-key-size-weak}
	\end{equation}
	then the minimum download cost $D_{\min}$ of the PIR system is given by 
	\begin{equation}
	D_{\min}=
	\begin{cases}
	L\cdot\left(\frac{N}{N-1}-\frac{1}{N-1}w\right),&\text{if }0\leq w\leq \frac{1}{N^{K-1}} \\
	D_{\min}^0,&\text{otherwise,}
	\end{cases}  \label{D-min-weak}
	\end{equation}
	where $D_{\min}^0$ is defined in \eqref{def-D-min-0}. 
	If $\rho<\rho_{\min}^{\text{w}}$, then $D_{\min}=\infty$. 
\end{theorem} 
%%%--------------------- END thm-tradeoff-weak ------------------------------
\begin{proof}
	The proof is given in \Cref{section-weak-proof}. 
\end{proof}
\begin{remark}
	If $\rho<\rho_{\min}^{\text{w}}$, it is impossible to meet all the requirements of the PIR system simultaneously. 
	In this case, the capacity $C=0$.
\end{remark}
\begin{remark}\label{remark-K=2}
	For $K=2$, the individual leakage is equal to the total leakage. The minimum download cost for the same value of individual and total leakage levels are the same, which can be easily verified from \eqref{D-min-strong} and \eqref{D-min-weak}. 
	Similarly, the minimum amount of common randomness are also equal to each other, i.e., $\rho_{\min}^{\text{s}}=\rho_{\min}^{\text{w}}$, which can be verified from \eqref{min-key-size-strong} and \eqref{min-key-size-weak}. 
	This can be seen from the coding scheme in \Cref{section-WS-PIR-code}. 
\end{remark}
\begin{remark}\label{remark-compare-same-leakage}
	If the amount of total leakage is $s'\triangleq \left(1+N+\cdots+N^{K-2}\right)\cdot w$, after a simple substitution of \eqref{D-min-strong},  the expression of minimum download cost in \eqref{D-min-strong} is equal to that in \eqref{D-min-weak}. Therefore, for the same download cost, the minimum total leakage is always $(1+N+\cdots+N^{K-2})$ times the minimum individual leakage.
\end{remark}

%In view of \eqref{min-key-size-weak} and \eqref{D-min-weak}, $D_{\min}$ shall be regarded as a bivariate function of $\rho$ and $w$. 
%For a given $w$, the dependency of $D_{\min}$ on $\rho$ only appears at the threshold $\rho_{\min}^{\text{w}}$ 
%so that i) $D_{\min}=\infty$ if $\rho<\rho_{\min}^{\text{w}}$; ii) $D_{\min}<\infty$ and $D_{\min}$ is independent of $\rho$ if $\rho\geq \rho_{\min}^{\text{w}}$. 

From \eqref{min-key-size-weak}, it is seen that $\rho_{\min}^{\text{w}}=0$ for $K\geq 3$, which is the case when one non-desired message can be used as the encryption key to protect another individual non-desired message. 
For $K=2$, there is only one non-desired message, and this makes it impossible to have another non-desired messages as the encryption key. In this case, $\rho_{\min}^{\text{w}}$ is a linear function of $w$ as shown in \eqref{min-key-size-weak}. 
For $\rho\geq \rho_{\min}^{\text{w}}$, the dependency of $D_{\min}$ on $(\rho,w)$ is illustrated in Fig.~\ref{fig-tradeoff-Dmin-rho-w} and Fig.~\ref{fig-tradeoff-weak}. 

%%------------------------ BEGIN-figure-tradeoff-strong ---------------------------------
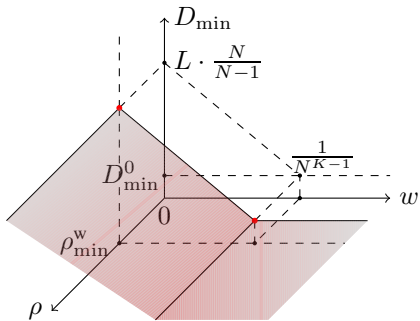
\begin{figure}[!ht]
	\centering
	\begin{tikzpicture}[scale=0.6]
	\draw[->] (0,0)--(0,4); \node [right] at (0,4) {$D_{\min}$};
	\draw[->] (0,0)--(5,0); \node [right] at (5,0) {$w$};
	\draw[->] (0,0)--(-2.5,-2.5); \node [left] at (-2.5,-2.5) {$\rho$};
	\node [below] at (0,0) {$0$};
	
	\draw (-1,2)--(-3.5,-0.5);  \draw[dashed] (-1,2)--(-1,3.6); 
	\draw (2,-0.5)--(-1,2);
	\draw (4.5,-0.5)--(2,-0.5)--(-0.2,-2.7);
	
	\draw[dashed] (-1,2)--(0,3)--(3,0.5)--(3,0); \filldraw (0,3) circle (1pt); \node [right] at (0,3) {$L\cdot\frac{N}{N-1}$}; 
	\draw[dashed] (0,0.5)--(3,0.5)--(5,0.5); \filldraw (3,0.5) circle (1pt); 
	\filldraw (0,0.5) circle (1pt); \node [left] at (0.1,0.5) {$D_{\min}^0$}; 
	\draw[dashed] (2,-0.5)--(3,0.5); \filldraw (3,0) circle (1pt); \node [above] at (3.5,0.35) {$\frac{1}{N^{K-1}}$}; 
	\draw[dashed] (-1,2)--(-1,-1); \filldraw (-1,-1) circle (1pt); \node [left] at (-1,-1) {$\rho_{\min}^{\text{w}}$}; 
	\draw[dashed] (-1,-1)--(2,-1)--(3,0) (2,-0.5)--(2,-1)--(4,-1); \filldraw (2,-1) circle (1pt); 
	
	\filldraw[red] (-1,2) circle (1.5pt); \filldraw[red] (2,-0.5) circle (1.5pt); 
	%	\fill[gray!20,nearly transparent] (-3.6,-0.6)--(-1,2)--(2,-0.5)--(-0.7,-3.2)--cycle; 
	%	\fill[gray!20,nearly transparent] (2,-0.5)--(-0.7,-3.2)--(1.8,-3.2)--(4.5,-0.5)--cycle; 
	\fill[top color=gray!10, bottom color=red,shading=axis,shading angle=45,opacity=0.15] (-3.5,-0.5)--(-1,2)--(2,-0.5)--(-0.2,-2.7)--cycle; 
	\fill[top color=red, bottom color=gray!10,shading=axis,shading angle=90,opacity=0.15] (2,-0.5)--(-0.2,-2.7)--(2.3,-2.7)--(4.5,-0.5)--cycle; 
	%	\node [right] at (2.5,1.3) {$D_{\min}^0\triangleq L\cdot\left(1+\frac{1}{N}+\cdots+\frac{1}{N^{K-1}}\right)$}; 
	\end{tikzpicture}
	\caption{The dependency of $D_{\min}$ on $(\rho,w)$.}
	\label{fig-tradeoff-Dmin-rho-w}
\end{figure}
%%-------------------------- END-figure-tradeoff-strong ---------------------------------
%The dashed lines project the corner points to the axises that show their values. 
%The value of $D_{\min}$ with respective to $(\rho,w)$ is given by the shaded area. 
%We see that for a given value of $w$, $D_{\min}$ is a constant and thus independent with $\rho$. 
%%For a given value of $\rho$ such that  $\rho \geq \rho_{\min}^{\text{w}}$, t
%The shaded area is projected onto the $D_{\min}$-$w$ plane, which is drawn in Fig.~\ref{fig-tradeoff-weak}. 

%%------------------------ BEGIN-figure-tradeoff-weak ---------------------------------
\begin{figure}[!ht]
	\centering
	\begin{tikzpicture}[scale=0.7]
	\draw[->] (0,0)--(0,2.5); \node [right] at (0,2.5) {$D_{\min}$};
	\draw[->] (0,0)--(3.5,0); \node [right] at (3.5,0) {$w$};
	\node [below] at (0,0) {$0$};
	
	\filldraw [red](0,1.5) circle (1pt); \node [left] at (0,1.5) {$L\cdot\frac{N}{N-1}$};
	\filldraw (2.0,0) circle (1pt); \node [below] at (2.0,0) {$\frac{1}{N^{K-1}}$}; \draw[dashed] (2.0,0)--(2.0,0.5);
	\filldraw (0,0.5) circle (1pt); \node [left] at (0,0.5) {$D_{\min}^0$}; \draw[dashed] (0,0.5)--(2.0,0.5);
	
	\draw[red] (0,1.5)--(2.0,0.5)--(3.3,0.5); \filldraw [red](2.0,0.5) circle (1pt);
	%	\fill[green!20] (0,2.5)--(0,1.5)--(2.0,0.5)--(3.5,0.5)--cycle;
	\end{tikzpicture}
	\caption{Tradeoff curve between $D_{\min}$ and $w$.}
	\label{fig-tradeoff-weak}
\end{figure}
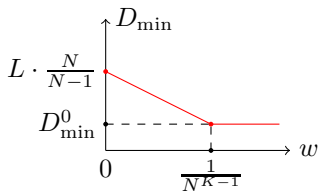
%%------------------------ BEGIN-figure-tradeoff-weak ---------------------------------

The red line in Fig.~\ref{fig-tradeoff-weak} is the tradeoff curve  between $D_{\min}$ and $w$ for $\rho\geq \rho_{\min}^{\text{w}}$. 
Similar to the observations of Fig.~\ref{fig-tradeoff-strong}, we see from Fig.~\ref{fig-tradeoff-weak} that for $w=0$, 
the problem becomes weakly secure PIR (WS-PIR) which differs from SPIR only in the type of security. 
The capacity of WS-PIR can be obtained by calculating $\frac{L}{D_{\min}}$ at $w=0$ from \eqref{D-min-weak}, 
which is given in the following corollary. 
\begin{corollary}\label{lemma-WS-PIR-capacity}
	If the amount of common randomness satisfies 
	\begin{equation}
	\rho\geq
	\begin{cases}
	\frac{1}{N-1},&\text{ if }K=2\\
	0,&\text{ if }K\geq 3, 
	\end{cases} \label{WS-PIR-key-size}
	\end{equation}
	the capacity of WS-PIR is given by 
	\begin{equation}
	C_{\text{WS-PIR}}=1-\frac{1}{N}.  \label{WS-PIR-capacity}
	\end{equation}
	If \eqref{WS-PIR-key-size} is not satisfied, the capacity $C_{\text{WS-PIR}}=0$. 
\end{corollary}
From this corollary, we observe that i) $C_{\text{WS-PIR}}=C_{\text{SPIR}}$, i.e., requiring strong security does not increase download cost compared to requiring weak security;
ii)  compared to weak security, we achieve strong security at the cost of a larger amount of common randomness. 
It is straightforward to verify that the optimal SPIR code in \cite{Sun-Jafar-SPIR-19IT} is also optimal for WS-PIR. 
In \Cref{section-WS-PIR-code}, we propose another optimal code for WS-PIR with $K\geq 3$, where the databases do not need to share common randomness.

\begin{remark}
	Consider a specialization of $s$ and $w$ between the corner point in Fig.~\ref{fig-tradeoff-strong} and Fig.~\ref{fig-tradeoff-weak}. 
	For an integer $\alpha\geq K-1$, we write $s$ and $w$ in the form of 
	\begin{equation}
	s=\frac{1}{N^{\alpha}}\cdot\frac{N^{K-1}-1}{N-1}   \label{specialization-strong}
	%=\frac{1}{N^\alpha}+\frac{1}{N^{\alpha-1}}+\cdots+\frac{1}{N^{\alpha-K+2}}
	\end{equation}
	and 
	\begin{equation}
	w=\frac{1}{N^\alpha}.   \label{specialization-weak}
	\end{equation}
	Denote the capacity of the PIR system for total and individual leakage constraints in \eqref{specialization-strong} and \eqref{specialization-weak} by $C_S$ and $C_W$, respectively. 
	Substituting \eqref{specialization-strong} into \Cref{thm-tradeoff-strong} and \eqref{specialization-weak} into \Cref{thm-tradeoff-weak}, 
	we obtain that for $\rho\geq \max\{\rho_{\min}^{\text{s}},\rho_{\min}^{\text{w}}\}$, 
	\begin{equation}
	C_S=C_W=\left(1+\frac{1}{N}+\frac{1}{N^2}+\cdots+\frac{1}{N^\alpha}\right)^{-1},  \label{capacity-strong-weak}
	\end{equation}
	which has a common form as $C_{\text{PIR}}=\big(1+\frac{1}{N}+\frac{1}{N^2}+\cdots+\frac{1}{N^{K-1}}\big)^{-1}$ and 
	$C_{\text{SPIR}}=C_{\text{WS-PIR}}=\big(1+\frac{1}{N}+\frac{1}{N^2}+\cdots+\frac{1}{N^{\infty}}\big)^{-1}$. 
	This demonstrates how the capacity decreases from  $C_{\text{PIR}}$ to $C_{\text{SPIR}}$ or $C_{\text{WS-PIR}}$ 
	while strengthening the security requirement (i.e., decrease of $s$ or $w$). 
	For $s$ and $w$ defined in \eqref{specialization-strong} and \eqref{specialization-weak}, when we increase $\alpha$, 
	the capacity of the system decreases by adding a term of $\frac{1}{N^\alpha}$ in the denominator of the capacity expression. 
	By increasing $\alpha$ from $K-1$ to $\infty$, $s$ (or $w$) decreases from the threshold $s_t$ (or $w_t$) to zero which means perfect security, 
	and the capacity finally decreases from $\left(1+\frac{1}{N}+\frac{1}{N^2}+\cdots+\frac{1}{N^{K-1}}\right)^{-1}$ to $\left(1+\frac{1}{N}+\frac{1}{N^2}+\cdots+\frac{1}{N^{\infty}}\right)^{-1}=1-\frac{1}{N}$. 
\end{remark}

%%======================================================================
\subsection{Minimum Message Size}\label{section-message-size}
The message size $L\log_2|\cX|$ is also an important factor to consider in practice~\cite{Tian-Sun-Chen-PIR-IT19}. 
The minimum message size is highly dependent on the download cost. 
%The extreme case of downloading everything can have message size of as small as 1, and
It was shown in \cite{Tian-Sun-Chen-PIR-IT19} that the message size of capacity-achieving PIR code is greater than or equal to $N-1$ 
if the code is uniformly decomposable. 
Before stating our results on the message size, we first define the notion of uniformly decomposable code, similar to that in~\cite{Tian-Sun-Chen-PIR-IT19}. 

A PIR code is called {\it decomposable}, if $\cY$ is a finite Abelian group, and for each $n\in[1:N]$ and $q\in\cQ_n$, 
the answer $A_n^{[k]}=\varphi_n(q,W_{1:K},S)$ can be written in the form 
\begin{align}
\varphi_n(q,W_{1:K},S)=\left(\varphi_{n,1}^{(q)}(W_{1:K},S),\cdots,\varphi_{n,\ell_n}^{(q)}(W_{1:K},S)\right)   \label{decompose-def-1}
\end{align}
where for $i\in\{1,2,\cdots,\ell_n\}$, 
\begin{align}
\varphi_{n,i}^{(q)}(W_{1:K},S)&=\varphi_{n,i,1}^{(q)}(W_{1})\oplus\cdots\oplus\varphi_{n,i,K}^{(q)}(W_{K})  \nonumber  \\
&\qquad\qquad\qquad\qquad \oplus\varphi_{n,i,K+1}^{(q)}(S),  \label{decompose-def-2}
\end{align}
where $\oplus$ denotes addition in the finite group $\cY$, 
and each $\varphi_{n,i,k}^{(q)}$ ($k\in[1:K]$) is a mapping $\cX^L\rightarrow\cY$ and $\varphi_{n,i,K+1}^{(q)}$ is a mapping $\cS\rightarrow\cY$. 

Furthermore, a decomposable code is called {\it uniform} if any component function $\varphi_{n,i,k}^{(q)}$ in \eqref{decompose-def-2} either satisfies the conditions 
\begin{align}
&\left|\left\{x\in\cX^L\!:\!\varphi_{n,i,k}^{(q)}(x)=y\right\}\right|=\left|\left\{x\in\cX^L\! :\! \varphi_{n,i,k}^{(q)}(x)=y'\right\}\right|,  \nonumber \\
&\hspace{4cm}  y,y'\in\cY,~k\in[1:K],  \label{uniform-def-1}
\end{align}
and
{\small\begin{align}
	&\left|\left\{x\in\cS:\varphi_{n,i,K+1}^{(q)}(x)=y\right\}\right|=\left|\left\{x\in\cS:\varphi_{n,i,K+1}^{(q)}(x)=y'\right\}\right|, \nonumber \\
	&\hspace{6cm} y,y'\in\cY,  \label{uniform-def-1.5}
	\end{align}}
or it maps everything to the same value, i.e., 
\begin{equation}
\varphi_{n,i,k}^{(q)}(x)=\varphi_{n,i,k}^{(q)}(x'),~x,x'\in\begin{cases}\cX^L,&\text{if }k\in[1:K]\\\cS,&\text{if }k=K+1.\end{cases}   \label{uniform-def-2}
\end{equation}

The terminology ``decomposable" refers to the property in \eqref{decompose-def-2} that each coded symbol is a summation of the component functions on each individual message. 
A uniformly decomposable code has an additional property that each mapping $\varphi_{n,i,k}^{(q)}$ either preserves the uniform probability distribution or induces a deterministic value. 
The notion of decomposable codes considerably generalizes the notion of linear codes. In particular, linear codes on finite fields are uniformly decomposable. 

The following lemma and theorem are our main results on the minimum message size. 
%%------------------------- BEGIN lemma-message-size --------------------------------
\begin{lemma}\label{lemma-message-size-multiple}
	The minimum message length of any Pareto optimal uniformly decomposable PIR code with either individual leakage constraint $w$ or total leakage constraint $s$ should be a multiple of $\log_{|\cX|}|\cY|$. 
\end{lemma}
\begin{proof}
	The proof can be found in Appendix \ref{section-message-size-multiple-proof}.
\end{proof}
%%------------------------- END lemma-message-size --------------------------------
\begin{remark}
	The lemma shows that the minimum $L$ is a multiple of $\log_{|\cX|}|\cY|$. However, the term $\log_{|\cX|}|\cY|$ is not necessarily an integer. 
\end{remark}

%In \Cref{thm-message-size}, we do not require the Pareto optimal codes have the minimum amount of common randomness.
We can see from the Pareto optimal code constructions in \Cref{section-strong-proof,section-weak-proof} that for most cases, we can simultaneously achieve a message length of $N-1$ and the minimum amount of common randomness characterized in \Cref{thm-tradeoff-strong,thm-tradeoff-weak}. 
The only exception is the individual leakage case with $K\geq 3$ and $|\cX|=N=2$, for which the message length is twice the minimum. 
Except for this case, we characterize the minimum message size in the following theorem. 
%%------------------------- BEGIN thm-message-size --------------------------------
\begin{theorem}\label{thm-message-size}
	Except for the case of $K\geq 3,|\cX|=N=2$, the minimum message size of any Pareto optimal uniformly decomposable PIR codes achieving the minimum amount of common random randomness with either individual leakage constraint $w$ or total leakage constraint $s$ is $(N-1)\log_2|\cY|$; in particular, it achieves $N-1$ when $|\cY|=|\cX|=2$. 
\end{theorem}
\begin{proof}
	Except for the case of $K\geq 3,|\cX|=N=2$, we have designed Pareto optimal codes achieving simultaneously 
	the minimum amount of common random randomness and the minimum message size. 
	Thus, we only need to prove the lower bound $L\log_2|\cX|\geq (N-1)\log_2|\cY|$, which is equivalent to $L\geq (N-1)\log_{|\cX|}|\cY|$. 
	The details can be found in Appendix~\ref{section-message-size-proof}.
\end{proof}
%%------------------------- END thm-message-size --------------------------------
\begin{remark}
	It is remarkable that the minimum message size of the Pareto optimal uniformly decomposable PIR codes with leakage constraints $w$ or $s$ is equal to $N-1$, which is the same as that of optimal uniformly decomposable classical PIR codes without security constraints. 
\end{remark}

%%======================================================================
\subsection{Relation to Shannon Cipher System}
In Shannon's imperfect secrecy system in \cite[pp. 71-72]{raymondbook}, 
$X$ is the plain text, $Y$ is the cipher text, and $Z$ is the encryption key. 
Since $X$ can be recovered from $Y$ and $Z$, we have $H(X|Y, Z) = 0$. 
It was proved in \cite{raymondbook} that this constraint implies
\begin{equation}
H(Z)\geq H(X)-I(X;Y).  \label{imperfect-key-size}
\end{equation}
This can be intuitively explained as follows: the key size is greater than or equal to amount of information that we need to protect, 
which is the difference between total amount of information $H(X)$ and the amount of information leakage $I(X;Y)$. 
The tradeoff between the minimum key size $H(Z)$ and the amount of information leakage $I(X;Y)$ is illustrated in Fig.~\ref{fig-comparison-imperfect}. 
%%------------------------ BEGIN-figure-tradeoff-weak ---------------------------------
\begin{figure}[!ht]
	\centering
	\begin{subfigure}{0.23\textwidth}
		\centering
		\begin{tikzpicture}[scale=0.65]
		\draw[->] (0,0)--(0,2.5); \node [right] at (0,2.5) {$H(Z)$};
		\draw[->] (0,0)--(3.0,0); \node [above] at (3.0,0) {$I(X;Y)$};
		\node [below] at (0,0) {$0$};
		
		\filldraw [red](0,1.5) circle (1pt); \node [left] at (0,1.5) {$H(X)$};
		\node [below] at (1.5,0) {$H(X)$};
		
		\draw[red] (0,1.5)--(1.5,0)--(2.5,0); \filldraw [red](1.5,0) circle (1pt);
		\end{tikzpicture}
		\caption{Tradeoff between $H(Z)$ and information leakage $I(X;Y)$.}
		\label{fig-comparison-imperfect}
	\end{subfigure}
	~
	\begin{subfigure}{0.23\textwidth}
		\centering
		\begin{tikzpicture}[scale=0.65]
		\draw[->] (0,0)--(0,2.5); \node [right] at (0,2.5) {$H(S)$};
		\draw[->] (0,0)--(3.0,0); \node [above] at (3.2,0) {$I(W_{\bar{k}};A_{1:N}^{[k]}\bF)$};
		\node [below] at (0,0) {$0$};
		
		\filldraw [red](0,1.5) circle (1pt); \node [left] at (0,1.5) {$\frac{1}{N-1}L$};
		\node [below] at (1.5,0) {$s_t L$};
		
		\draw[red] (0,1.5)--(1.5,0)--(2.5,0); \filldraw [red](1.5,0) circle (1pt);
		\end{tikzpicture}
		\caption{Tradeoff between $H(S)$ and total leakage $I(W_{\bar{k}};A_{1:N}^{[k]}\bF)$.}
		\label{fig-comparison-PIR}
	\end{subfigure}
	\caption{Comparison of encryption key size between imperfect secrecy system and PIR system}
	\label{fig-comparison}
\end{figure}
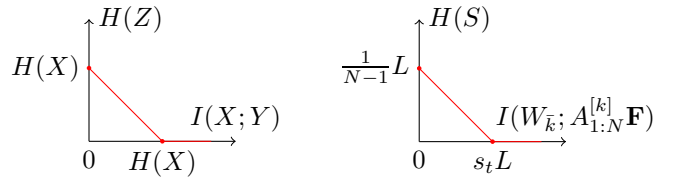
%%------------------------ END-figure-tradeoff-weak ---------------------------------

Our result on the lower bound of common randomness in \eqref{min-key-size-strong} of \Cref{thm-tradeoff-strong} presents the similar form as \eqref{imperfect-key-size}, 
which can be written as 
\begin{equation}
H(S)\geq \frac{1}{N-1}L-\frac{N^{K-1}}{N^{K-1}-1}I(W_{\bar{k}};A_{1:N}^{[k]}\bF). \label{imperfect-PIR-size}
\end{equation}
This can be interpreted as follows, where we focus on one particular answer. 
%From the optimal code in \Cref{section-strong-achievability}, we see that 
The average size of any individual answer should be at least $\frac{1}{N-1}L$, 
and the term $\frac{N^{K-1}}{N^{K-1}-1}I(W_{\bar{k}};A_{1:N}^{[k]}\bF)$ in \eqref{imperfect-PIR-size} can roughly be viewed as the average amount of information leakage per answer. 
The amount of common randomness must be greater than or equal to the amount of information we need to protect,
which is the difference between the total amount of information in an answer $\frac{1}{N-1}L$ and the information leakage of that answer $\frac{N^{K-1}}{N^{K-1}-1}I(W_{\bar{k}};A_{1:N}^{[k]}\bF)$. 
The tradeoff between $H(S)$ and $I(W_{\bar{k}};A_{1:N}^{[k]}\bF)$ is illustrated in Fig.~\ref{fig-comparison-PIR}. 
The tradeoff curve is the same as that in Fig.~\ref{fig-comparison-imperfect}, and thus our result can also be interpreted as a generalization of Shannon's cipher system to multiple users with a privacy requirement.

%It is easy to verify that $H(S)$ and $I(W_{\bar{k}};A_{1:N}^{[k]}Q_{1:N}^{[k]}\bF)$ at the two corner points reduce to that of the classical PIR and SPIR, respectively. 

%%=============================================================================================
%\section{Essential Proofs}\label{section-proof}
%%=============================================================================================
\section{Proof of \Cref{thm-tradeoff-strong}}\label{section-strong-proof}
We prove \Cref{thm-tradeoff-strong} by showing the achievability, converse, and the minimum key size in the following. 
%%=============================================================================================
\subsection{Achievability}\label{section-strong-achievability}
Following the analysis of Fig.~\ref{fig-tradeoff-strong}, we only need to prove the theorem for $s\in[0,s_t]$. 
Consider a message length of $L=N-1$. 
In addition to the query generation in \eqref{query-generation}, now the random key $\bF$ generates one more indicator bit $F_0\in\{0,1\}$, 
according to probability $P(F_0=0)=\frac{N^{K-1}(N-1)s}{N^{K-1}-1}\in[0,1]$ and $P(F_0=1)=1-\frac{N^{K-1}(N-1)s}{N^{K-1}-1}$. 
Note that the probability of $F_0$ is chosen according to the coding strategy so that the information leakage is within the constraint $s$. 

For $x,y\in[1:N]$, it is useful to define the operation $(x+y)_N$ by 
\begin{equation}
(x+y)_N=
\begin{cases}
x+y,& \text{if }x+y\leq N  \\
x+y-N,& \text{if }x+y>N.
\end{cases} \label{def-addition}
\end{equation}
And similarly, 
\begin{equation}
(x-y)_N=
\begin{cases}
x-y,& \text{if }x-y>0  \\
x-y+N,& \text{if }x-y\leq 0.
\end{cases}  \label{def-subtraction}
\end{equation}
For random key $\bF=(F_0,F_1,F_2,\cdots,F_{K-1})\in\{0,1\}\times[1:N]^{K-1}$, let $F^*\triangleq \left(\sum_{i=1}^{K-1}F_i\right)_N$. 
Then the query for the $n$-th database is generated as 
\begin{align}
&Q_n^{[k]}=\left(F_0,F_1,\cdots,F_{k-1},(n-F^*)_N,F_k,\cdots,F_{K-1}\right),  \nonumber \\
&\hspace{4cm}n\in[1:N],k\in[1:K]. \label{strong-code-query}
\end{align}
Since the query is a length-$(K+1)$ vector, we can denote it by $Q_{n,0:K}^{[k]}$. 
Upon receiving the queries, the databases generate the answers using the TSC code for $F_0=0$, and SPIR code for $F_0=1$. 
Specifically, the answer is 
\begin{equation}
A_n^{[k]}=
\begin{cases}
W_{1,Q_{n,1}^{[k]}}\oplus \cdots \oplus W_{K,Q_{n,K}^{[k]}},&\hspace{-0.2cm}\text{if }Q_{n,0}^{[k]}=0\\
W_{1,Q_{n,1}^{[k]}}\oplus \cdots \oplus W_{K,Q_{n,K}^{[k]}}\oplus S,&\hspace{-0.2cm}\text{if }Q_{n,0}^{[k]}=1,
\end{cases}
\end{equation}
where $S$ is the common randomness shared among all databases. 
The retrieval and privacy requirements are easily seen from the TSC and SPIR codes. 
Then we can check the code performances as follows. 
\begin{enumerate}
	\item Information leakage: in view of the performance of TSC code in \eqref{TSC-performance} and \eqref{TSC-total-leakage}, 
	the overall normalized total leakage is 
	\begin{align}
	&\frac{1}{L}I(W_{\bar{k}};Q_{1:N}^{[k]},A_{1:N}^{[k]})  \nonumber \\
	&=\frac{N^{K-1}(N-1)s}{N^{K-1}-1}\cdot\frac{1}{N-1}\left(1-\frac{1}{N^{K-1}}\right)  \nonumber \\
	&\quad +\left[1-\frac{N^{K-1}(N-1)s}{N^{K-1}-1}\right]\cdot 0   \\
	&=s,
	\end{align}
	which is the total leakage constraint. 
	
	\item Download cost: by \eqref{TSC-performance} and \eqref{SPIR-performance}, we have 
	\begin{align}
	D&=\frac{N^{K-1}(N-1)s}{N^{K-1}-1}\cdot\frac{N^K-1}{N^{K-1}}  \nonumber \\
	&\qquad +\left[1-\frac{N^{K-1}(N-1)s}{N^{K-1}-1}\right]\cdot N  \\
	&=N-\frac{N-1}{N^{K-1}-1}s  \\
	&=L\cdot\left(\frac{N}{N-1}-\frac{1}{N^{K-1}-1}s\right),  \label{code-strong-TSC-download}
	\end{align}
	which is the minimum download cost in \eqref{D-min-strong}. 
	
	\item Amount of common randomness: since common randomness is only used in the SPIR code, we derive from \eqref{SPIR-performance} that  
	\begin{align}
	\rho=\frac{1-\frac{N^{K-1}(N-1)s}{N^{K-1}-1}}{L}  =\frac{1}{N-1}-\frac{N^{K-1}}{N^{K-1}-1} s,  \label{code-strong-TSC-keysize}
	\end{align}
	which is equal to $\rho_{\min}^{\text{s}}$ in \eqref{min-key-size-strong} of \Cref{thm-tradeoff-strong}. 
\end{enumerate}
This proves the achievability of \Cref{thm-tradeoff-strong} for $s\in[0,s_t]$. 
\begin{remark}
	The code is obtained by combining TSC code and SPIR code with certain probabilities, 
	which is accomplished by sending one more query bit as an indicator. 
	This is different from time-sharing which uses sub-codes at a fraction of time. 
	Our method i) takes advantage of the fact that the two codes have the same message size, 
	and the message size of the combined code remains the same, which is $N-1$ here; 
	ii) achieves the desired leakage and download cost in expectation, while time-sharing achieves the same information leakage and download cost in a deterministic manner; 
	iii) outperforms the time-sharing approach in the sense that time-sharing may require a large message size, 
	which limits the applicability of the code, while the proposed scheme always uses a message size of $N-1$. 
\end{remark}

%%===========================================================
\subsection{Converse}\label{section-strong-converse}
Similar to the converse proof in \cite{Sun-Jafar-PIR-capacity-2017IT}, we first lower-bound the entropy of answers and messages iteratively by the message length $L$. 
Then the download cost $D$ and information leakage~$s$ is proved to be greater than or equal to this lower bound, which gives the relation of $D$, $s$, and $L$. 
The converse in \cite{Sun-Jafar-PIR-capacity-2017IT} lower-bounds $D$ by $L$, while the converse proof here needs to lower-bound $D$ in terms of both $L$ and~$s$. 
In a sense, this is a more general version of the proof in \cite{Sun-Jafar-PIR-capacity-2017IT}. 

We first present an iterative lemma that will be used later. 
\begin{lemma}\label{lemma-strong-iteration}
	For all $k\in[1:K]$,
	\begin{equation*}
	NH(A_{1:N}^{[k]},W_{1:k}|\bF)\geq H(A_{1:N}^{[k+1]},W_{1:k+1}|\bF)+k(N-1)L.
	\end{equation*}
\end{lemma}
\begin{proof}
	See Appendix \ref{lemma-strong-iteration-proof}. 
\end{proof}
Applying \Cref{lemma-strong-iteration} repeatedly for $k=1,2,\cdots,K-1$, we have 
\begin{align}
&N^{K-1}H(A^{[1]}_{1:N}W_1|\bF)  \nonumber \\
&\geq N^{K-2}H(A^{[2]}_{1:N}W_{1:2}|\bF)+N^{K-2}(N-1)L   \label{converse-s-iteration00} \\
&\geq N^{K-3}H(A^{[3]}_{1:N}W_{1:3}|\bF)  \nonumber \\
&\qquad +N^{K-2}(N-1)L+2N^{K-3}(N-1)L  \\
&\qquad \vdots \\
&\geq H(A^{[K]}_{1:N}W_{1:K}|\bF)+\left(\sum_{i=1}^{K-1}iN^{K-1-i}\right)(N-1)L  \\
&= H(A^{[K]}_{1:N}W_{1:K}|\bF)+\left(\sum_{i=1}^{K-1}N^i-(K-1)\right)L. \label{converse-s-iteration0}
\end{align}
By symmetry, for any $k,k'\in[1:K]$ and $k\neq k'$, we have 
\begin{align}
N^{K-1}H(A^{[k]}_{1:N}W_k|\bF)&\geq H(A^{[k']}_{1:N}W_{1:K}|\bF)  \nonumber \\
&\quad +\left(\sum_{i=1}^{K-1}N^i-(K-1)\right)L. \label{converse-s-iteration1}
\end{align}
To obtain a lower bound on download cost $D=\log_{|\cX|}|\cY|\sum_{n=1}^N\bE(\ell_n)$ (c.f. \eqref{def-D}), consider 
\begin{align}
D&=\log_{|\cX|}|\cY|\sum_{n=1}^N\bE(\ell_n)  \nonumber \\
&=\sum_{q_{1:N}\in\cQ_{1:N}}\Pr(Q_{1:N}^{[k]}=q_{1:N})\sum_{n=1}^N\ell_n \log_{|\cX|}|\cY|    \label{converse-download0}  \\
&\geq \sum_{q_{1:N}\in\cQ_{1:N}}\Pr(Q_{1:N}^{[k]}=q_{1:N})\sum_{n=1}^NH(A_n^{(q_n)})    \\
&\geq \sum_{q_{1:N}\in\cQ_{1:N}}\Pr(Q_{1:N}^{[k]}=q_{1:N})H(A_1^{(q_1)},\cdots,A_N^{(q_N)})    \\
&=\sum_{q_{1:N}\in\cQ_{1:N}}\!\!\!\Pr(Q_{1:N}^{[k]}=q_{1:N})H(A_{1:N}^{[k]}|Q_{1:N}^{[k]}=q_{1:N})    \\
&=H(A_{1:N}^{[k]}|\bF).  \label{converse-download}
\end{align}
Then we derive the following, 
\begin{align}
&\left(N^{K-1}-1\right)D+I(W_{\bar{k}}; A^{[k]}_{1:N} Q^{[k]}_{1:N} \bF)  \nonumber \\
&\geq \left(N^{K-1}-1\right)H(A_{1:N}^{[k]}|\bF)+I(W_{\bar{k}};A^{[k]}_{1:N}|\bF)   \label{converse-s-D-1}  \\
&=\left(N^{K-1}-1\right)H(A_{1:N}^{[k]}|\bF)  \nonumber \\
&\quad +\left[H(W_{\bar{k}})+H(A^{[k]}_{1:N}W_k|\bF)-H(A^{[k]}_{1:N}W_{1:K}|\bF)\right]   \\
&=\left[(N^{K-1}-1)H(A_{1:N}^{[k]}|\bF)+H(A^{[k]}_{1:N}W_k|\bF)\right] \nonumber \\
&\quad +(K-1)L-H(A^{[k]}_{1:N}W_{1:K}|\bF)  \\
&=N^{K-1}H(A^{[k]}_{1:N}W_k|\bF)+(K-1)L-H(A^{[k]}_{1:N}W_{1:K}|\bF)   \label{converse-s-D-4} \\
&\geq \left[H(A^{[k']}_{1:N}W_{1:K}|\bF)+\left(\sum_{i=1}^{K-1}N^i-(K-1)\right)L\right] \nonumber \\
&\qquad +(K-1)L-H(A^{[k]}_{1:N}W_{1:K}|\bF)  \label{converse-s-D-5} \\
&=\left(\sum_{i=1}^{K-1}N^i\right)L+\left[H(A^{[k']}_{1:N}W_{1:K}|\bF)-H(A^{[k]}_{1:N}W_{1:K}|\bF)\right],  \label{converse-s-D}
\end{align}
where \eqref{converse-s-D-4} follows from 
\begin{equation}
H(A_{1:N}^{[k]}|\bF)=H(A^{[k]}_{1:N} Q^{[k]}_{1:N}W_k|\bF)=H(A^{[k]}_{1:N}W_k|\bF),
\end{equation}
and \eqref{converse-s-D-5} follows from \eqref{converse-s-iteration1} for any $k'\in[1:K]$ and $k'\neq k$. 
Similarly, by switching $k$ and $k'$ in \eqref{converse-s-D}, we obtain 
\begin{align}
&\left(N^{K-1}-1\right)D+I(W_{\bar{k'}}; A^{[k']}_{1:N} Q^{[k']}_{1:N} \bF) \nonumber \\
&\geq\left(\sum_{i=1}^{K-1}N^i\right)L+\left[H(A^{[k]}_{1:N}W_{1:K}|\bF)-H(A^{[k']}_{1:N}W_{1:K}|\bF)\right].  \label{converse-s-D'}
\end{align}
Summing up \eqref{converse-s-D} and \eqref{converse-s-D'}, we have 
\begin{align}
&2\left(N^{K-1}-1\right)D+I(W_{\bar{k}}; A^{[k]}_{1:N} \bF)+I(W_{\bar{k'}}; A^{[k']}_{1:N} \bF) \nonumber \\
&\geq 2\left(\sum_{i=1}^{K-1}N^i\right)L.  
\end{align}
Thus the lower bound on $D$ is 
\begin{align}
D&\geq \frac{N}{N-1}L \nonumber \\
&~~-\frac{1}{N^{K-1}-1}\frac{I(W_{\bar{k}}; A^{[k]}_{1:N}\bF)+I(W_{\bar{k'}}; A^{[k']}_{1:N} \bF)}{2}  \\
&\geq L\cdot\left(\frac{N}{N-1}-\frac{1}{N^{K-1}-1}\cdot s\right),  \label{converse-final-bound-D}
\end{align}
which matches the value of $D_{\min}$ for $s\in[0,s_t]$ in \eqref{D-min-strong} of \Cref{thm-tradeoff-strong}. 
For $s>s_t$, the converse follows from that of classical PIR. 
This completes the converse proof.

\begin{remark}
	The proof steps for \Cref{lemma-strong-iteration} are similar to those for Lemma 6 in~\cite{Sun-Jafar-PIR-capacity-2017IT}, 
	and \eqref{converse-s-iteration00}-\eqref{converse-s-iteration1} are similar to (62)-(67) in~\cite{Sun-Jafar-PIR-capacity-2017IT}, which complete the converse proof therein by connecting $D$ to $L$. 
%	The lower bound of $D$ in \eqref{converse-download0}-\eqref{converse-download} can be found in~\cite{Tian-Sun-Chen-PIR-IT19}. 
	The key difference between the proof here and that in \cite{Sun-Jafar-PIR-capacity-2017IT} is in steps \eqref{converse-s-D-1}-\eqref{converse-final-bound-D}, which connects $D$ to $L$ and $s$ jointly. 
\end{remark}

%%==========================================================
\subsection{Amount of common randomness}\label{section-strong-key-size}
To obtain the lower bound on the amount of common randomness, we first consider the following,
\begin{align}
&N^{K-1}I(W_{\bar{k}};A_{1:N}^{[k]}Q_{1:N}^{[k]}\bF)+(N^{K-1}-1)H(S)   \nonumber \\
&=N^{K-1}\left[H(W_{\bar{k}})+H(A^{[k]}_{1:N}W_k|\bF)-H(A^{[k]}_{1:N}W_{1:K}|\bF)\right] \nonumber \\
&\qquad +(N^{K-1}-1)H(S) \\
&=N^{K-1}H(W_{\bar{k}})+\!\!\left[N^{K-1}H(A^{[k]}_{1:N}W_k|\bF)\!-\!H(A^{[k]}_{1:N}W_{1:K}|\bF)\right]  \nonumber \\
&\qquad +(N^{K-1}-1)\left[H(S)-H(A^{[k]}_{1:N}W_{1:K}|\bF)\right] \\
&\geq N^{K-1}(K-1)L-(N^{K-1}-1)H(W_{1:K})  \nonumber \\
&\qquad +\left[N^{K-1}H(A^{[k]}_{1:N}W_k|\bF)-H(A^{[k]}_{1:N}W_{1:K}|\bF)\right] \label{converse-s-key-size-2}  \\
&\geq (K-N^{K-1})L+\left(\sum_{i=1}^{K-1}N^i-(K-1)\right)L  \nonumber \\ 
&\qquad +\left[H(A^{[k']}_{1:N}W_{1:K}|\bF)-H(A^{[k]}_{1:N}W_{1:K}|\bF)\right]  \label{converse-s-key-size-3}  \\
&=\left(\sum_{i=0}^{K-2}N^i\right)L+\left[H(A^{[k']}_{1:N}W_{1:K}|\bF)-H(A^{[k]}_{1:N}W_{1:K}|\bF)\right],    \label{converse-s-key-size-5}
\end{align}
where \eqref{converse-s-key-size-2} follows from 
\begin{align}
&H(A^{[k]}_{1:N}W_{1:K}|\bF)-H(W_{1:K})=H(A^{[k]}_{1:N}|W_{1:K}\bF)  \\
&=H(A^{[k]}_{1:N}|W_{1:K}Q^{[k]}_{1:N}\bF)-H(A^{[k]}_{1:N}|W_{1:K}Q^{[k]}_{1:N}\bF S)  \\
&=I(S;A^{[k]}_{1:N}|W_{1:K}Q^{[k]}_{1:N}\bF). 
\end{align}
and 
\begin{equation}
H(S)-I(S;A^{[k]}_{1:N}|W_{1:K}Q^{[k]}_{1:N}\bF)\geq 0, 
\end{equation}
and \eqref{converse-s-key-size-3} follows from \eqref{converse-s-iteration1} for any $k'\in[1:K]$ and $k'\neq k$. 
Similarly, by switching $k$ and $k'$ in \eqref{converse-s-key-size-5}, we obtain 
\begin{align}
&N^{K-1}I(W_{\bar{k'}};A_{1:N}^{[k']}Q_{1:N}^{[k']}\bF)+(N^{K-1}-1)H(S)   \nonumber \\
&\geq \left(\sum_{i=0}^{K-2}N^i\right)L+\left[H(A^{[k]}_{1:N}W_{1:K}|\bF)-H(A^{[k']}_{1:N}W_{1:K}|\bF)\right].    \label{converse-s-key-size-5'}
\end{align}
Summing up \eqref{converse-s-key-size-5} and \eqref{converse-s-key-size-5'}, we have  
{\small\begin{align}
	&2(N^{K-1}-1)H(S)+N^{K-1}\left[I(W_{\bar{k}};A_{1:N}^{[k]}\bF)+I(W_{\bar{k'}};A_{1:N}^{[k']}\bF)\right]   \nonumber \\
	&\geq 2\left(\sum_{i=0}^{K-2}N^i\right)L.
	\end{align}}
Then we have 
\begin{align}
\rho&=\frac{H(S)}{L} \geq \frac{\sum_{i=0}^{K-2}N^i}{N^{K-1}-1}L  \nonumber \\
&\quad -\frac{N^{K-1}}{N^{K-1}-1}\frac{I(W_{\bar{k}};A_{1:N}^{[k]}\bF)+I(W_{\bar{k'}};A_{1:N}^{[k']}\bF)}{2}   \\
&\geq \frac{1}{N-1}-\frac{N^{K-1}}{N^{K-1}-1}\cdot s,
\end{align}
which is exactly the lower bound in \eqref{min-key-size-strong} of \Cref{thm-tradeoff-strong}.

%%=========================================================
\section{Proof of \Cref{thm-tradeoff-weak}}\label{section-weak-proof}
%We prove \Cref{thm-capacity-weak} by showing the achievability and the converse in the following. 
%The time sharing between optimal PIR code and WS-PIR code will be optimal in terms of achieving the tradeoff curve in Fig.~\ref{fig-tradeoff-weak}.
%Thus, before proving \Cref{thm-tradeoff-weak}, we first propose an optimal WS-PIR code that achieves the minimum amount of common randomness. 
%The code is also the achievability proof for \Cref{lemma-WS-PIR-capacity}. 
%%=========================================================
\subsection{Optimal code for WS-PIR}\label{section-WS-PIR-code}
Even though the SPIR code in \cite{Sun-Jafar-SPIR-19IT} achieves the capacity of WS-PIR,
it does not always achieve the minimum amount of common randomness for WS-PIR. 
From \eqref{WS-PIR-key-size} in \Cref{lemma-WS-PIR-capacity}, we see that for $K=2$, the minimum amount of common randomness of WS-PIR is equal to that of SPIR. 
This is because for $K=2$, weak security is equivalent to strong security which means that WS-PIR is equivalent to SPIR. 
Thus, the SPIR code in \cite{Sun-Jafar-SPIR-19IT} also achieves the minimum amount of common randomness of WS-PIR. 

For $K\geq 3$, \eqref{WS-PIR-key-size} indicates that the databases do not need to share common randomness. 
We consider the following three cases: 
\begin{enumerate}[i)]
	\item $|\cX|\geq 3$ and $N\geq 2$;
	\item $|\cX|=2$ and $N\geq 3$;
	\item $|\cX|=N=2$. 
\end{enumerate}
Next, we propose an optimal code for WS-PIR that uses the sum of all message symbols as encryption key shared by databases and no extra randomness is needed. 
%The code is designed separately for two cases i) $N\geq 3$; ii) $N=2$. %Consider $K\geq 3$ and $N\geq 3$. 
The code design is simply to modify TSC code by adding the shared encryption key (sum of all message symbols) to each of the answers. 

\textit{Case i):} The code has a message length of $L=N-1$. 
Specifically, by appending dummy variables $W_{k,N}=0$, the message $W_k$ can be written as 
\begin{equation}
W_k=(W_{k,1},W_{k,2}\cdots,W_{k,N-1},W_{k,N}). 
\end{equation}
Let the random key $\bF$ of the user be chosen from $[1:N]^{K-1}$ which gives 
\begin{equation}
\bF=(F_1,F_2,\cdots,F_{K-1}).
\end{equation}
%where $F_k\in\cX$ for $k=1,2,\cdots,K-1$. 
For $x,y\in[1:N]$, the operations $(x+y)_N$ and $(x-y)_N$ are defined by \eqref{def-addition} and \eqref{def-subtraction}.
Let $F^*\triangleq \left(\sum_{i=1}^{K-1}F_i\right)_N$. 
For $k\in[1:K]$ and $n\in[1:N]$, the query is a deterministic function of the random key $\bF$, defined as 
\begin{equation}
Q_n^{[k]}=\big(F_1,F_2,\cdots,F_{k-1},(n-F^*)_N,F_k,\cdots,F_{K-1}\big).  \label{WS-PIR-code-query}
\end{equation}
The sum of all message symbols is denoted by $S$, which is 
\begin{equation}
S=\sum_{k=1}^{K}\big(W_{k,1}\oplus W_{k,2}\oplus\cdots\oplus W_{k,N-1}\big).  \label{WS-PIR-code-encryption-key}
\end{equation}
Upon receiving the query $Q_n^{[k]}$, the $n$-th database generates an answer $A_n^{[k]}$ using $Q_n^{[k]}$ as linear combination indexes of all the message symbols. 
We further add the encryption key $S$ to each answer and obtain that 
\begin{align}
A_n^{[k]}&=W_{1,F_1}\oplus\cdots \oplus W_{k-1,F_{k-1}}\oplus W_{k,(n-F^*)_N}  \nonumber \\
&\qquad \qquad \oplus W_{k+1,F_k}\oplus\cdots\oplus W_{K,F_{K-1}}\oplus S.   \label{WS-PIR-code-ans-allsymbols}
\end{align}
For simplicity, we define 
\begin{equation}
B\!=\! W_{1,F_1}\oplus \cdots \oplus\! W_{k-1,F_{k-1}}\!\oplus\! W_{k+1,F_k}\!\oplus\cdots\oplus W_{K,F_{K-1}},  \label{B-def}
\end{equation}
where $\oplus$ denotes the addition operation in finite group $\cX$ (this is an abuse of notation without ambiguity). 
Substituting \eqref{B-def} into \eqref{WS-PIR-code-ans-allsymbols}, we have 
\begin{equation}
A_n^{[k]}=W_{k,(n-F^*)_N}\oplus B\oplus S.    \label{WS-PIR-code-ans}
\end{equation}
The user receives all the answers from databases, i.e., $A_1^{[k]},A_2^{[k]},\cdots,A_N^{[k]}$.
Then we see that 
\begin{align}
W_{k,(n-F^*)_N}=A_n^{[k]}\ominus A_{F^*}^{[k]}=A_n^{[k]}\ominus (B\oplus S),  \label{WS-PIR-code-recover}
\end{align}
where $\ominus$ is the subtraction operation in the Abelian group $\cX$. 
Then the message $W_k$ can be recovered by ranging $n$ from $1$ to $N$. 

Since $\bF$ is chosen uniformly from $\{0,1\}^{K-1}$ and all the queries are deterministic functions of $\bF$, 
we see that $Q_n^{[k]}$ is chosen uniformly from the query set for any $k\in[1:K]$ and $n\in[1:N]$. 
Thus $Q_n^{[k]}$ provides no information about the message index $k$, which ensures the privacy. 

The weak security (individual leakage) can be seen as follows. 
%	In addition to $W_k$, the user can recover from databases only $B\oplus S$. 
The coefficient of each message symbol $W_{k,i}~(k\in[1:K],i\in[1:N-1])$ in the expression of $B\oplus S$ can only be $1$ or 2. 
Because of the assumption that $|\cX|\geq 3$, none of the message symbols will vanish. 
Since $K\geq 3$, there is at least one $k''\in[1:K]$ such that $k''\neq k,k'$. 
Then the message symbols of $W_{k''}$ randomizes the sum of symbols from $W_k$ and $W_{k'}$, and thus we obtain that 
\begin{equation}
I(B\oplus S;W_{k'}W_k)=0,  \label{WS-PIR-code-security-pre}
\end{equation}
which implies  
\begin{equation}
I(B\oplus S;W_k)=I(B\oplus S;W_{k'}|W_k)=0. 
\end{equation}
Thus, we have 
\begin{equation}
I(W_{k'};B\oplus S,W_k)=I(W_{k'};W_k)+I(W_{k'};B\oplus S|W_k)=0.  \label{WS-PIR-code-security-0}
\end{equation}
To see the security, we consider the following, 
\begin{align}
I(W_{k'};A_{1:N}^{[k]} Q_{1:N}^{[k]} \bF) &=I(W_{k'};A_{1:N}^{[k]} Q_{1:N}^{[k]} W_k \bF)   \label{WS-PIR-code-security-1} \\
&=I(W_{k'};A_{1:N}^{[k]} W_k\bF)   \\
&=I(W_{k'};B\oplus S,W_k|\bF)   \label{WS-PIR-code-security-2}  \\
%&=\sum_{\mathbf{f}}p(\mathbf{f})I(W_{k'};B\oplus S,W_k|\bF=\mathbf{f})  \nonumber \\
&= 0,  \label{WS-PIR-code-security-4}
\end{align}
where \eqref{WS-PIR-code-security-1} follows from the recovery in \eqref{WS-PIR-code-recover}, 
\eqref{WS-PIR-code-security-2} follows from the expression of answer in \eqref{WS-PIR-code-ans} and $I(W_{k'};\bF)=0$, 
and the last equality follows from \eqref{WS-PIR-code-security-0} and $B\oplus S$ is a function of messages $W_{1:K}$ and the random key $\bF$.
%\begin{align}
%I(W_{k'};A_{1:N}^{[k]} Q_{1:N}^{[k]} \bF) &=I(W_{k'};A_{1:N}^{[k]} Q_{1:N}^{[k]} W_k \bF)   \label{WS-PIR-code-security-1} \\
%&=I(W_{k'};A_{1:N}^{[k]} W_k\bF)   \\
%&=I(W_{k'};B\oplus S, W_k\bF)   \label{WS-PIR-code-security-2} \\
%&=I(W_{k'};B\oplus S|W_k\bF)   \label{WS-PIR-code-security-3} \\
%&=0,  \label{WS-PIR-code-security-4}
%\end{align}
%where \eqref{WS-PIR-code-security-1} follows from the recovery in \eqref{WS-PIR-code-recover}, 
%\eqref{WS-PIR-code-security-2} follows from the expression of answer in \eqref{WS-PIR-code-ans}, 
%\eqref{WS-PIR-code-security-3} follows from the independence of $W_{k'}$ and $(W_k,\bF)$, 
%and the last step follows from \eqref{WS-PIR-code-security-0} and the independence of $\bF$ and $(W_k,W_{k'})$. 

Since each answer consists of one symbol, the download cost is 
\begin{equation}
D=N. \label{D-WS-PIR}
\end{equation}
Since $L=N-1$, the rate of the code is simply $R=\frac{N-1}{N}=1-\frac{1}{N}$ that matches the capacity in \eqref{WS-PIR-capacity} of \Cref{lemma-WS-PIR-capacity}. 
Lastly, the databases do not share common randomness in addition to the messages themselves. 

\textit{Case ii):} It is easy to verify that the code also works for $|\cX|=2$ and $N\geq 3$. 
This can be seen intuitively as follows. Since $L=N-1\geq 2$, each message consists of at least two symbols. 
The encryption key $S$ consists of all the message symbols. 
Since $W_{k,(n-F^*)_N}\oplus B$ contains at most one symbol from each message, 
it can cancel out at most one symbol from $S$. 
Then the answer $A_n^{[k]}=W_{k,(n-F^*)_N}\oplus B\oplus S$ (c.f. \eqref{WS-PIR-code-ans}) contains at least one symbol from each message. 
Since $K\geq 3$, \eqref{WS-PIR-code-security-pre} and the subsequent equations also hold for the current case. 
This ensures the security of each individual message. 

\textit{Case iii):} For the case that $N=2$ and $|\cX|=2$, let $L=2(N-1)=2$. 
We can divide each message into two sub-messages, i.e., $W_k=(W_k^{(1)},W_k^{(2)})$ for $k\in[1:K]$. 
Then we apply the code in Case i) to each sub-message, where we choose the same encryption key as $S=\sum_{k=1}^KW_k^{(1)}\oplus W_k^{(2)}$. 
The parameters  $\bF^{(1)},\bF^{(2)}\in\{1,2\}^{K-1}$ are chosen to be dual of each other, i.e., $\bF^{(1)}$ and $\bF^{(2)}$ differ in every position. 
The private retrieval requirement is easy. 
We can see the weak security from the following example of $K=3$. 
Let $k=2$ and the random key to be chosen as $\bF^{(1)}=(2,1)$ and $\bF^{(2)}=(1,2)$, then the queries are
\begin{equation}
q_1^{(1)}=(2,1,1),~q_2^{(1)}=(2,2,1)
\end{equation}
and
\begin{equation}
q_1^{(2)}=(1,1,2),~q_2^{(2)}=(1,2,2).
\end{equation}
The parameters $B^{(1)}\oplus S$ and $B^{(2)}\oplus S$ of the two sub-codes are thus 
\begin{align}
\!B^{(1)}\!\oplus\! S&\!=\! W_{1,2}^{(1)}\oplus W_{3,1}^{(1)}\oplus S\!=\! W_1^{(1)}\oplus W_1^{(2)}\oplus W_3^{(2)},   \\
\!B^{(2)}\!\oplus\! S&\!=\! W_{1,1}^{(2)}\oplus W_{3,2}^{(2)}\oplus S\!=\! W_1^{(1)}\oplus W_3^{(1)}\oplus W_3^{(2)}.
\end{align}
We can easily see that $(B^{(1)}\oplus S, B^{(2)}\oplus S)$ provides no information about either $W_1$ or $W_3$. 
Then using similar arguments as \eqref{WS-PIR-code-security-pre}-\eqref{WS-PIR-code-security-4}, we can obtain the weak security.

\begin{remark}
	The SPIR code in \cite{Sun-Jafar-SPIR-19IT} is an optimal code for WS-PIR that achieves capacity. 
	%	Additionally, the message length and download cost here are both equal to that of the SPIR code. 
	%	In fact, simply replacing the common randomness $S$ in \cite{Sun-Jafar-SPIR-19IT} by the sum of all message symbols 
	%	will also be an optimal WS-PIR code that achieves minimum amount of common randomness. 
	However, it is important to point out that the method of using the sum of message symbols as encryption keys does not work for SPIR, since it leaks partial of the total information of the whole message set $W_{1:K}$. 
\end{remark}

%%=========================================================
\subsection{Achievability of Pareto Optimal Points}\label{section-weak-achievability}
Following the analysis of Fig.~\ref{fig-tradeoff-weak} after \Cref{lemma-WS-PIR-capacity}, we only need to prove the theorem for $w\in[0,w_t]$. 
Since for $K=2$, the WS-PIR code reduces to SPIR code and the achievability is proved in \Cref{section-strong-achievability}, we consider only $K\geq 3$ here. 

Similar to \Cref{section-strong-achievability}, consider a message length of $L=N-1$. 
The random key $\bF$ generates one more indicator bit $F_0\in\{0,1\}$ according to probability $P(F_0=0)=wN^{K-1}\in[0,1]$ and $P(F_0=1)=1-wN^{K-1}$, which is chosen to control the individual leakage within the constraint $w$. 
The queries and answers are generated similarly as that in \Cref{section-strong-achievability}. 
The only difference is to replace $S$ by the encryption key using \eqref{WS-PIR-code-encryption-key}. 
We then check the following. 
\begin{enumerate}
	\item Information leakage: in view of the performance of TSC code in \eqref{TSC-performance} and \eqref{TSC-individual-leakage}, 
	the overall normalized total leakage is 
	\begin{align}
	&\frac{1}{L}I(W_{\bar{k}};Q_{1:N}^{[k]},A_{1:N}^{[k]},\bF)  \nonumber \\
	&=wN^{K-1}\cdot\frac{1}{N^{K-1}}+(1-wN^{K-1})\cdot 0   \\
	&=w,
	\end{align}
	which is the individual leakage constraint. 
	
	\item Download cost: by \eqref{TSC-performance} and \eqref{D-WS-PIR}, we have 
	\begin{align}
	D&=wN^{K-1}\cdot\frac{N^K-1}{N^{K-1}}+(1-wN^{K-1})\cdot N  \\
	&=N-w  \\
	&=L\cdot\left(\frac{N}{N-1}-\frac{1}{N-1}w\right),  \label{code-weak-TSC-download}
	\end{align}
	which is the minimum download cost in \eqref{D-min-weak} of \Cref{thm-tradeoff-weak}.
	
	\item Amount of common randomness: i) for $K=2$, the key size in \eqref{code-strong-TSC-keysize} reduces to $\frac{1}{N-1}-\frac{N}{N-1}\cdot w$; 
	ii) for $K\geq 3$, there is no common randomness needed and $\rho=0$. 
	Thus the amount of common randomness of the code is equal to $\rho_{\min}^{\text{s}}$ in \eqref{min-key-size-strong} of \Cref{thm-tradeoff-strong}. 
	
\end{enumerate}
This proves the achievability of \Cref{thm-tradeoff-weak} for $w\in[0,w_t]$. 

%\begin{remark}
%	It is worth noting that the time sharing between PIR code and SPIR code in \cite{Sun-Jafar-SPIR-19IT} is also optimal that achieves the whole tradeoff curve in Fig.~\ref{fig-tradeoff-weak}, 
%	since SPIR code achieves the capacity of WS-PIR. 
%	However, the code needs common randomness to be shared among databases, which is less efficient. 
%\end{remark}
\begin{remark}
	We notice that the linear combination coefficients as well as the other code performances here are the same as that of the linear combination between TSC and SPIR codes in \Cref{section-strong-achievability}. The reason is $C_{\text{WS-PIR}}=C_{\text{PIR}}$ and the WS-PIR code in \Cref{section-WS-PIR-code} has the same message length and download cost as SPIR code in \cite{Sun-Jafar-SPIR-19IT}.  
\end{remark}

\begin{remark}
	The achievability proof of \Cref{thm-tradeoff-strong} in \Cref{section-strong-achievability} is a probabilistic combination of two existing codes, i.e., the TSC code in \cite{Tian-Sun-Chen-PIR-IT19} and SPIR code in \cite{Sun-Jafar-SPIR-19IT}.
	However, the achievability proof of \Cref{thm-tradeoff-weak} in \Cref{section-weak-achievability} uses a new WS-PIR code proposed in \Cref{section-WS-PIR-code} that does not require common randomness. 
\end{remark}

%%=============================================================================================
\subsection{Converse}\label{section-weak-converse}
The converse proof is very similar to that in \Cref{section-strong-converse} except that the iteration is replaced by a single relation between the entropy of answers for desired and non-desired messages.
		
For any $k,k'\in [1:K]$ and $k \not= k'$, we first derive the following which is similar to the proof of \Cref{lemma-strong-iteration}.
\begin{align}
&NH(A_{1:N}^{[k]}W_k|\bF)\geq \sum_{n=1}^{N}H(A_n^{[k]}W_k|\bF)  \label{converse-w-iteration-1} \\
&=\sum_{n=1}^{N}H(A_n^{[k]}W_k|Q_n^{[k]}) \label{converse-w-iteration-2} \\
&=\sum_{n=1}^{N}H(A_n^{[k']}W_k|Q_n^{[k']})  \label{converse-w-iteration-3} \\
&=\sum_{n=1}^{N}H(A_n^{[k']}W_k|\bF)  \label{converse-w-iteration-4} \\
&=\sum_{n=1}^{N}H(A_n^{[k']}|W_k\bF)+NH(W_k)  \\
&\geq \sum_{n=1}^{N}H(A_n^{[k']}|W_k\bF A_{1:n-1}^{[k']})+NL  \label{converse-w-iteration-5} \\
&=H(A_{1:N}^{[k']}|W_k\bF)+NL  \\
&= H(A_{1:N}^{[k']}W_k|\bF)+(N-1)L  \\
&= H(A_{1:N}^{[k']}Q_{1:N}^{[k']}W_kW_{k'}|\bF)+(N-1)L  \\
&= H(A_{1:N}^{[k']}W_kW_{k'}|\bF)+(N-1)L, 
\end{align}
where \eqref{converse-w-iteration-2} and \eqref{converse-w-iteration-4} follow from the Markov chain $(A_n^{[k]},W_{[1:K]})\rightarrow Q_n^{[k]}\rightarrow\bF$, 
and \eqref{converse-w-iteration-3} follows from identical distribution constraint in \eqref{privacy-identical-distribute}. 
Thus, we have 
\begin{align}
NH(A_{1:N}^{[k]}W_k|\bF)\!-\!H(A_{1:N}^{[k']}W_kW_{k'}|\bF)\geq (N-1)L.  \label{converse-w-iteration}
\end{align}
To provide a lower bound on the download cost $D$ (c.f. \eqref{def-D}), we derive the following 
\begin{align}
&(N-1)D+I(W_{k'};A^{[k]}_{1:N}Q^{[k]}_{1:N} \bF)  \nonumber \\
&\geq (N-1)H(A_{1:N}^{[k]}|\bF)+I(W_{k'};A^{[k]}_{1:N}|\bF)   \label{converse-weak-proof-1} \\
&= \left[H(W_{k'})+H(A_{1:N}^{[k]}W_k|\bF)-H(A_{1:N}^{[k]}W_kW_{k'}|\bF)\right]  \nonumber \\
&\qquad + (N-1)H(A_{1:N}^{[k]}W_k|\bF)  \label{converse-weak-proof-2} \\
&=NH(A_{1:N}^{[k]}W_k|\bF)-H(A_{1:N}^{[k]}W_kW_{k'}|\bF)+L   \label{converse-weak-proof-3} \\
&\geq NL\!+\!\left[H(A_{1:N}^{[k']}W_kW_{k'}|\bF)\!-\!H(A_{1:N}^{[k]}W_kW_{k'}|\bF)\right]   \label{converse-weak-proof-4}
\end{align}
where \eqref{converse-weak-proof-1} follows from \eqref{converse-download}, 
\eqref{converse-weak-proof-2} follows from \eqref{query-generation} and \eqref{recovery-constraint}, 
and \eqref{converse-weak-proof-4} follows from \eqref{converse-w-iteration}. 
By switching $k$ and $k'$ in  \eqref{converse-weak-proof-4}, we obtain 
\begin{align}
&(N-1)D+I(W_{k};A^{[k']}_{1:N}Q^{[k']}_{1:N} \bF)  \nonumber \\
&\geq NL\!+\!\left[H(A_{1:N}^{[k]}W_kW_{k'}|\bF)\!-\!H(A_{1:N}^{[k']}W_kW_{k'}|\bF)\right].   \label{converse-weak-proof-4'}
\end{align}
Summing up \eqref{converse-weak-proof-4} and \eqref{converse-weak-proof-4'}, we have  
\begin{align*}
2(N-1)D+\left[I(W_{k'};A^{[k]}_{1:N}\bF)+I(W_{k};A^{[k']}_{1:N}\bF)\right]\geq 2NL. 
\end{align*}
Thus, the lower bound on $D$ is 
\begin{align}
D&\geq \frac{N}{N-1}L -\frac{1}{N-1}\frac{I(W_{k'};A^{[k]}_{1:N}\bF)+I(W_{k};A^{[k']}_{1:N}\bF)}{2}   \nonumber \\
&\geq L\cdot\left(\frac{N}{N-1}-\frac{1}{N-1} w\right)
\end{align}
which matches the value of $D_{\min}$ for $w\in[0,w_t]$ in \eqref{D-min-weak} of \Cref{thm-tradeoff-weak}. 
For $w>w_t$, the converse follows from that of classical PIR. 
This completes the converse proof. 

%%==================================================================
\subsection{Minimum Key Size}\label{section-weak-key-size}
For $K=2$, individual leakage is equal to total leakage, which means that PIR system with the same value of total leakage constraint $s$ and individual leakage constraint $w$ are equivalent. 
This gives the minimum amount of common randomness in \eqref{min-key-size-weak} of \Cref{thm-tradeoff-weak}. 
For $K\geq 3$, the lower bound in \eqref{min-key-size-weak} is zero, which is trivial.

\section{Conclusion}\label{section-conclusion}
In this work, we studied the PIR problem with total and individual leakage constraints, respectively. 
Our main contribution was the characterization of the tradeoffs between the minimum download cost $D_{\min}$ and the total leakage constraint $s$ and individual leakage constraint $w$.
The minimum amount of common randomness with respect to $s$ and $w$ was also characterized. 
Remarkably, for $K\geq 3$, the databases do not need to share common randomness in order to achieve the individual leakage constraint $w$. 
%It was shown that the linear combination of classical PIR code (e.g., SJ code in \cite{Sun-Jafar-PIR-capacity-2017IT} and TSC code in \cite{Tian-Sun-Chen-PIR-IT19}) and i) SPIR code in \cite{Sun-Jafar-SPIR-19IT}; ii) WS-PIR code are both optimal in terms of achieving the corresponding tradeoff curve. 
For $s\in[0,s_t]$, it was shown that the linear combination of TSC code and SPIR code is Pareto optimal (achieving the whole tradeoff curve). 
For $w\in[0,w_t]$, an optimal code that achieves the capacity of WS-PIR without shared common randomness was proposed ($K\geq 3$), 
and its linear combination with the classical PIR code was proved Pareto optimal. 
The proposed Pareto optimal codes for both individual and total leakage were proved to have a minimum message size of $N-1$.

%%==================================================================
\appendices
\section{Proof of \Cref{lemma-message-size-multiple}}\label{section-message-size-multiple-proof}
Throughout the proof, we use $L$ to denote the minimum message length for Pareto optimal uniformly decomposable codes. 
Since we have coding scheme of message length $(N-1)\log_{|\cX|}|\cY|$, the minimum message length should be $L\leq (N-1)\log_{|\cX|}|\cY|$. 
If $L= (N-1)\log_{|\cX|}|\cY|$, we are done. Otherwise, $L< (N-1)\log_{|\cX|}|\cY|$. 

In the following, we prove the lemma by contradiction and suppose that $\frac{L}{\log_{|\cX|}|\cY|}$ is not an integer. 
We only consider the individual leakage case, since by \Cref{remark-compare-same-leakage} the same results also hold for the total leakage case. 
From \eqref{D-min-weak} in \Cref{thm-tradeoff-weak}, the average number of downloaded symbols in $\cY$ for Pareto optimal codes should be 
\begin{equation}
D_{\cY}=\frac{L}{\log_{|\cX|}|\cY|}\left(\frac{N}{N-1}-\frac{w}{N-1}\right)=\frac{N-w}{N-1}\frac{L}{\log_{|\cX|}|\cY|}.  \label{min-siz-multiple-pf-Dy}
\end{equation}
Since $L< (N-1)\log_{|\cX|}|\cY|$ implies $0<\frac{N-w}{N-1}\frac{L}{\log_{|\cX|}|\cY|}-\frac{L}{\log_{|\cX|}|\cY|}<1$, 
the only possible realization of queries that have a downloaded number of symbols $d\leq D_{\cY}$ are those $q$ inducing $d=\lceil\frac{L}{\log_{|\cX|}|\cY|}\rceil$ downloaded symbols. 
It is easy to see that 
\begin{equation}
L< \lceil\frac{L}{\log_{|\cX|}|\cY|}\rceil\cdot\log_{|\cX|}|\cY|<(L+1).  \label{min-siz-multiple-pf-bound-L}
\end{equation}
Since the equality in \eqref{converse-weak-proof-1} is equivalent to equality in \eqref{converse-download}, 
we see from Appendix \ref{section-message-size-proof} that independence property \textbf{P}1 holds for both total and individual leakage cases. 
By \textbf{P}1 and uniform distribution of answer symbols, we obtain 
\begin{equation}
H(A_{1:N}^{(q)})=\lceil\frac{L}{\log_{|\cX|}|\cY|}\rceil\cdot \log_{|\cX|}|\cY|.  \label{min-siz-multiple-pf-H_A}
\end{equation}
%which implies that 
%\begin{equation}
%H(A_{1:N}^{(q)}|W_k)\geq H(A_{1:N}^{(q)})-H(W_k)>0.   \label{min-siz-multiple-pf-1}
%\end{equation}
Since the whole answer set can recover $W_k$, we have $H(W_k|A_{1:N}^{(q)})=0$, and thus 
$H(A_{1:N}^{(q)}|W_k)= H(A_{1:N}^{(q)})-H(W_k).$
Then from \eqref{min-siz-multiple-pf-bound-L} and \eqref{min-siz-multiple-pf-H_A}, we can bound $H(A_{1:N}^{(q)}|W_k)$ by 
\begin{equation}
0<H(A_{1:N}^{(q)}|W_k)<\log_{|\cX|}|\cY|.   \label{min-siz-multiple-pf-1}
\end{equation}
Next, we derive another bound for $H(A_{1:N}^{(q)}|W_k)$, which leads to a contradiction. 
By the definition of uniform code in \eqref{uniform-def-1}-\eqref{uniform-def-2}, denote the non-deterministic symbols of $A_{1:N}^{(q)}$ by $A_{n,i}^{(q)}=\varphi_{n,i}^{(q)}$. 
For $n\in[1:N]$ and $i\in[1:\ell_n]$, define a random variable by $U_{n,i}^{(q)}=A_{n,i}^{(q)}\ominus\varphi_{n,i,k}^{(q)}(W_k)$, where $\ominus$ is the subtraction in the Abelian group $\cX$ as defined after \eqref{WS-PIR-code-recover}. Then, we have 
\begin{align}
H(A_{1:N}^{(q)}|W_k)=H(U_{1,1}^{(q)},\cdots,U_{N,\ell_{N}}^{(q)}).  \label{min-siz-multiple-pf-condition-k}
\end{align}
We see that $U_{n,i}^{(q)}$ is either deterministic or uniformly distributed over~$\cY$, thus $H(U_{n,i}^{(q)})=0$ or $\log_{|\cX|}|\cY|$. 
If $H(U_{n,i}^{(q)})=0$ for all $(n,i)$, we have $H(U_{1,1}^{(q)},\cdots,U_{N,\ell_{N}}^{(q)})=0$. 
Otherwise, there is at least one pair  $(n,i)$ such that $H(U_{n,i}^{(q)})=\log_{|\cX|}|\cY|$, 
which implies that $H(U_{1,1}^{(q)},\cdots,U_{N,\ell_{N}}^{(q)})\geq \log_{|\cX|}|\cY|$. 
Thus, we see from \eqref{min-siz-multiple-pf-condition-k} that 
\begin{equation}
H(A_{1:N}^{(q)}|W_k)\begin{cases}=0\\\geq \log_{|\cX|}|\cY|.\end{cases}   \label{min-siz-multiple-pf-2}
\end{equation}
Now, \eqref{min-siz-multiple-pf-2} contradicts with \eqref{min-siz-multiple-pf-1}. 
Then the supposition that $\frac{L}{\log_{|\cX|}|\cY|}$ is non-integer is false, and thus $L$ is a multiple of $\log_{|\cX|}|\cY|$. 

%%==================================================================
\section{Proof of \Cref{thm-message-size}} \label{section-message-size-proof}
\subsection{Total Leakage Case}
In \cite{Tian-Sun-Chen-PIR-IT19}, three properties P1-P3 were used to prove the minimum message size. 
Here, we verify that P1-P3 are also true for Pareto optimal uniformly decomposable PIR codes with total leakage constraint $s$, 
and then the minimum message size can be obtained accordingly. 
We first define some notations before presenting the properties. 
For notational simplicity, we write $W_{K+1}=S$ and then $W_{1:K+1}=(W_{1:K},S)$ and $W_{\bar{k}}=(W_1,\cdots,W_{k-1},W_{k+1},\cdots,W_K,S)$. 
For $n\in[1:N]$ and $q\in\cQ_n$, the answer for decomposable codes can be simply written in matrix form as 
\begin{equation}
A_n^{[k]}=\varphi_n(Q_n^{[k]}=q,W_{1:K},S)=W_{1:K+1}\cdot G_n^{(q)},  \label{msg-siz-pf-ans1}
\end{equation} 
where $W_{1:K+1}$ is viewed as a length-$(K+1)$ vector, 
and $G_n^{(q)}$ is a matrix of dimension $(K+1)\times l_n$ whose elements $G_{n,i,k}^{(q)}$ are functions $\cX^L(\text{or }\cS)\rightarrow\cY$ with the ``$\cdot$" operation defined as 
$W_k\cdot G_{n,i,k}^{(q)}\triangleq \varphi_{n,i,k}^{(q)}(W_k).$
Let $q_{1:N}$ be a set of queries such that $\Pr(Q_{1:N}^{[k]}=q_{1:N})>0$. 
For any $\mathcal{A}\subseteq[1:K]$, let $G_{n|\mathcal{A}}^{(q)}$ be the submatrix of $G_n^{(q)}$ with rows corresponding to $\{W_i,i\in\mathcal{A}\}$ removed. 
For Pareto optimal uniformly decomposable PIR codes, the three properties are stated as follows. 
\begin{enumerate}[{\bf P}1.]
	\item The $N$ random variables $A_1^{(q_1)},A_2^{(q_2)},\cdots,A_N^{(q_N)}$ are mutually independent, 
	where $A_n^{(q_n)}$ is the $n$-th answer when $Q_n^{[k]}=q_n$.
	
	\item The $N$ random variables $W_{\bar{k}}\cdot G_{1|k}^{(q_1)}, W_{\bar{k}}\cdot G_{2|k}^{(q_2)},\cdots,W_{\bar{k}}\cdot G_{N|k}^{(q_N)}$ are deterministic of each other. 
	
	\item The $N$ random variables $W_{k}\cdot G_{1|\bar{k}}^{(q_1)}, W_{k}\cdot G_{2|\bar{k}}^{(q_2)},\cdots,W_{k}\cdot G_{N|\bar{k}}^{(q_N)}$ are mutually independent. 
\end{enumerate}
Then the properties can be verified as follows. 
\begin{enumerate}[i)]
	\item It is easily seen from the proof of \textbf{P}1 in \cite{Tian-Sun-Chen-PIR-IT19} that 
	the equality of \eqref{converse-download} in the current paper holds if and only if the equality in Lemma 1 of \cite{Tian-Sun-Chen-PIR-IT19} holds. 
	
	\item When \eqref{converse-s-iteration-1} is equality, we obtain for $k=1$ that 
	\begin{equation}
	H(A_{1:N}^{[1]}|W_{1}A_n^{[1]}\bF)=0,
	\end{equation}
	which is exactly the condition that proves \textbf{P}2 in \cite{Tian-Sun-Chen-PIR-IT19}. 
	%	Thus, if the equality in \eqref{converse-s-iteration-0} holds, P2 is true. 
	
	\item The inequality in \eqref{converse-s-iteration-5} is exactly the same as (d) in the proof of Lemma 2 in \cite{Tian-Sun-Chen-PIR-IT19}, which induces \textbf{P}3. 
\end{enumerate}
The above observations ensure that the properties \textbf{P}1-\textbf{P}3 are also true for Pareto optimal uniformly decomposable PIR codes with total leakage constraint $s$. 
Then we can use the argument in \cite{Tian-Sun-Chen-PIR-IT19} to lower bound the message length as $L\geq (N-1)\log_{|\cX|}|\cY|$.

\subsection{Individual Leakage Case}
For the individual leakage case, the following lemma provides a new method of proving the minimum message size. 
We can also apply the same method to the total leakage case, where the analysis of total leakage may be more complicated. The details for verifying this is omitted here. 
%are also true for uniformly decomposable capacity-achieving PIR codes with either total leakage level $s$ or individual leakage level $w$. 
\begin{lemma}\label{lemma-independence-property}
	The message length of Pareto optimal uniformly decomposable PIR codes with individual leakage constraint $w$ is lower bounded by $L\geq (N-1)\log_{|\cX|}|\cY|$. 
\end{lemma}
\begin{proof}
	For $w>\frac{1}{N^{K-1}}$, the problem is equivalent to classical PIR. 
	Thus, we only consider $w\in[0,\frac{1}{N^{K-1}}]$, for which we prove the lemma by contradiction. 
	By \Cref{lemma-message-size-multiple}, we suppose there exist Pareto optimal uniformly decomposable codes such that 
	\begin{equation}
	L\leq (N-2)\log_{|\cX|}|\cY|.  \label{msg-siz-pf-assumption}
	\end{equation}
	Fix a $k\in[1:K]$. For each query realization $Q_{1:N}^{[k]}=q_{1:N}$, let $d$ be the corresponding total number of downloaded symbols in $\cY$, which can be written as $d=\sum_{n=1}^N\ell_n(q_n). $ 
	Because of recovery requirement of $W_k$, we have 
	\begin{equation}
	d\geq \frac{L}{\log_{|\cX|}|\cY|}.  \label{msg-siz-pf-d>=}
	\end{equation}
	For nonnegative integer $m$, let $\cQ^{(m)}$ be the set of query realizations defined by 
	\begin{equation}
	\cQ^{(m)}\triangleq \left\{q_{1:N}:d=\frac{L}{\log_{|\cX|}|\cY|}+m\right\}. 
	\end{equation}
	Define a probability $P^{(m)}$ by 
	\begin{equation}
	P^{(m)}=\sum_{q_{1:N}\in \cQ^{(m)}}\Pr(q_{1:N}). 
	\end{equation}
	From \eqref{min-siz-multiple-pf-Dy}, the average number of downloaded symbols in $\cY$ for Pareto optimal codes is 
	\begin{equation}
	D_{\cY}=\frac{N-w}{N-1}\frac{L}{\log_{|\cX|}|\cY|},  \label{msg-siz-pf-def-D_Y}
	\end{equation}
	which, by the supposition \eqref{msg-siz-pf-assumption} and the range of $w\in[0,\frac{1}{N^{K-1}}]$, satisfies 
	$0<D_{\cY}-\frac{L}{\log_{|\cX|}|\cY|}<1, $
	and thus $D_{\cY}$ is not an integer. 
	In light of the lower bound of $d$ in \eqref{msg-siz-pf-d>=}, the only possible realization of queries that induce $d<D_{\cY}$ are $q_{1:N}\in \cQ^{(0)}$, 
	and the corresponding downloaded number of symbols is $d=\frac{L}{\log_{|\cX|}|\cY|}$. 
	Then we have 
	\begin{align}
	D_{\cY}&=\sum_{m}P^{(m)}\left(\frac{L}{\log_{|\cX|}|\cY|}+m\right)  \\
	&=\left(\sum_{m}P^{(m)}\right)\frac{L}{\log_{|\cX|}|\cY|}+\sum_{m\geq 1}mP^{(m)}  \\
	&=\frac{L}{\log_{|\cX|}|\cY|}+\sum_{m\geq 1}mP^{(m)}  \\
	&\geq \frac{L}{\log_{|\cX|}|\cY|}+\sum_{m\geq 1}P^{(m)}  \\
	&=\frac{L}{\log_{|\cX|}|\cY|}+(1-P^{(0)}),
	\end{align}
	which by \eqref{msg-siz-pf-def-D_Y} implies that 
	\begin{equation}
	P^{(0)}\geq 1-\frac{1-w}{N-1}\frac{L}{\log_{|\cX|}|\cY|}.  \label{msg-siz-pf-probability0}
	\end{equation}
	%	By the definition of uniformly decomposable codes in \eqref{decompose-def-1}-\eqref{uniform-def-2}, 
	%	the answer $A_n^{[k]}$ consists of $\ell_n$ uniformly distributed random variables over $\cY$. 
	
	In the following, we will derive a lower bound for information leakage that leads to a contradiction with the individual leakage constraints. 
	For $q_{1:N}\in \cQ^{(0)}$, we have 
	\begin{equation}
	d=\ell_1+\ell_2+\cdots+\ell_{N}=\frac{L}{\log_{|\cX|}|\cY|}.   \label{msg-siz-pf-length}
	\end{equation}
	which implies $\ell_n\in\left[0:\frac{L}{\log_{|\cX|}|\cY|}\right]$ for any $n\in[1:N]$. 
	Then for any $n\in[1:N]$, define a parameter by 
	\begin{equation}
	T_n^{[k]}\triangleq \sum_{i=1}^{\frac{L}{\log_{|\cX|}|\cY|}}i\cdot \Pr\left\{q_{1:N}\in \cQ^{(0)}, \ell_n=i\right\}, 
	\end{equation}
	which can be roughly viewed as the average length of the answer when $q_{1:N}\in \cQ^{(0)}$. 
	Then it is intuitive that 
	$\sum_{n=1}^N T_n^{[k]}=\frac{L}{\log_{|\cX|}|\cY|}\cdot P^{(0)},$
	%	\begin{align}
	%	\sum_{n=1}^N T_n^{[k]}&=\sum_{n=1}^N\sum_{i=1}^{\frac{L}{\log_{|\cX|}|\cY|}}i\cdot \Pr\left\{q_{1:N}\in \cQ^{(0)}, \ell_n=i\right\}   \\
	%	&=\left(\sum_{n=1}^N\sum_{i=1}^{\frac{L}{\log_{|\cX|}|\cY|}}i\cdot \Pr\left\{\ell_n=i \big| q_{1:N}\in \cQ^{(0)}\right\}\right)\cdot\Pr\left\{q_{1:N}\in \cQ^{(0)}\right\}   \\
	%	&=\frac{L}{\log_{|\cX|}|\cY|}\cdot P^{(0)},
	%	\end{align}
	which implies that there exists an $n\in[1:N]$ such that 
	\begin{align}
	T_n^{[k]}&\geq \frac{\frac{L}{\log_{|\cX|}|\cY|}}{N}\cdot P^{(0)}.
	\end{align}
	In the following, we consider this particular $n$. 
	Then by \eqref{msg-siz-pf-probability0}, we have 
	\begin{equation}
	T_n^{[k]}\geq \frac{\frac{L}{\log_{|\cX|}|\cY|}}{N}\left(1-\frac{1-w}{N-1}\frac{L}{\log_{|\cX|}|\cY|}\right).  \label{msg-siz-pf-prob-Ank} 
	\end{equation}
	For $q_{1:N}\in \cQ^{(0)}$, since each answer symbol is either deterministic or uniformly distributed over $\cY$, we obtain from \eqref{msg-siz-pf-length} that 
	\begin{equation}
	H(A_{1:N}^{(q_{1:N})})\leq \frac{L}{\log_{|\cX|}|\cY|}\cdot\log_{|\cX|}|\cY|=L=H(W_k). 
	\end{equation}
	This together with the recovery requirement of $W_k$, i.e., $H(W_k|A_{1:N}^{(q_{1:N})})=0$, imply that 
	$H(A_{1:N}^{(q_{1:N})})=H(W_k). $
	Then for the queries $q_{1:N}\in \cQ^{(0)}$, the answers $A_{1:N}^{(q_{1:N})}$ and the message $W_k$ are deterministic of each other, 
	and all the answer symbols are non-deterministic. 
	Let $A_n^{(q_n)}=\left(A_{n,1}^{(q_n)},A_{n,2}^{(q_n)},\cdots,A_{n,\ell_n}^{(q_n)}\right)$ where each component $A_{n,j}^{(q_n)}\in\cY$. 
	Then we have 
	\begin{equation}
	I(W_k;A_{n,j}^{(q_n)})=H(A_{n,j}^{(q_n)})=\log_{|\cX|}|\cY|, \forall j\in[1:\ell_n].   \label{msg-siz-pf-leakage-unit}
	\end{equation}
	Define a set of query realizations of $Q_n^{[k]}$ by $\cQ_n^{(0)}\triangleq \left\{q_n: q_{1:N}\in \cQ^{(0)}\right\}$. 
	Since for $q_{1:N}\notin \cQ^{(0)}$, the $n$-th query can also be in $\cQ_n^{(0)}$, we have 
	\begin{equation}
	\Pr\{q_n\in\cQ_n^{(0)}\}\geq \Pr\{q_{1:N}\in \cQ^{(0)}\}.   \label{msg-siz-pf-query-probabilities}
	\end{equation}
	
	Next, we use \eqref{msg-siz-pf-prob-Ank} to derive an lower bound for individual leakage that will lead to a contradiction. 
	Consider $k'\neq k$. Similarly, we define $T_n^{[k']}$  by 
	\begin{equation}
	T_n^{[k']}\triangleq \sum_{i=1}^{\frac{L}{\log_{|\cX|}|\cY|}}i\cdot \Pr\left\{q_n\in\cQ_n^{(0)}, \ell_n=i\right\}. 
	\end{equation}
	Then the leakage of $W_k$ from the answers for $k'$ can be bounded as follows, 
	\begin{equation}
	I(W_k;A_{1:N}^{[k']}Q_{1:N}^{[k']})\geq I(W_k;A_n^{[k']}Q_n^{[k']})\geq T_n^{[k']}\cdot \log_{|\cX|}|\cY|,  \label{msg-siz-pf-leakage-bound}
	\end{equation}
	where the last inequality follows from \eqref{msg-siz-pf-leakage-unit}. 
	Because of identical distribution of $Q_n^{[k]}$ and $Q_n^{[k']}$ in \eqref{privacy-identical-distribute}, we obtain from \eqref{msg-siz-pf-query-probabilities} that  
	\begin{equation*}
	\Pr\left\{Q_n^{[k']}\!\in \! \cQ_n^{(0)}\right\}=\Pr\left\{Q_n^{[k]} \! \in \! \cQ_n^{(0)}\right\}\geq \Pr\left\{Q_{1:N}^{[k]} \!\in\! \cQ^{(0)}\right\}, 
	\end{equation*}
	which implies that 
	\begin{equation}
	T_n^{[k']}\geq T_n^{[k]}. \label{msg-siz-pf-prob-Ank'-iid}
	\end{equation}
	By \eqref{msg-siz-pf-leakage-bound}, \eqref{msg-siz-pf-prob-Ank'-iid}, and the lower bound in \eqref{msg-siz-pf-prob-Ank}, we obtain 
	\begin{equation}
	I(W_k;A_{1:N}^{[k']}Q_{1:N}^{[k']})\geq \frac{L}{N}\left(1-\frac{1-w}{N-1}\frac{L}{\log_{|\cX|}|\cY|}\right).   \label{msg-siz-pf-leakage}
	\end{equation}
	
	In the following, we derive the contradiction by arguing 
	\begin{equation}
	\frac{L}{N}\left(1-\frac{1-w}{N-1}\frac{L}{\log_{|\cX|}|\cY|}\right)>wL,  \label{msg-siz-pf-inequality-0}
	\end{equation}
	which can be verified equivalent to 
	\begin{equation}
	w<\frac{(N-1)-\frac{L}{\log_{|\cX|}|\cY|}}{(N^2-N)-\frac{L}{\log_{|\cX|}|\cY|}}.   \label{msg-siz-pf-inequality-1}
	\end{equation}
	It is easy to see that the right-hand side is a decreasing function of $L$. 
	Since by assumption $L\leq (N-2)\frac{L}{\log_{|\cX|}|\cY|}$, we have 
	\begin{align}
	\frac{(N-1)-\frac{L}{\log_{|\cX|}|\cY|}}{(N^2-N)-\frac{L}{\log_{|\cX|}|\cY|}}&\geq \frac{(N-1)-(N-2)}{(N^2-N)-(N-2)}  \nonumber \\
	&=\frac{1}{N^2-2(N-1)}. \label{msg-siz-pf-inequality-2}
	\end{align}
	For $K\geq 3$ and $N\geq 2$, we can see that $\frac{1}{N^2-2(N-1)}>\frac{1}{N^{K-1}}$.
	%	\begin{equation}
	%	\frac{1}{N^2-2(N-1)}>\frac{1}{N^{K-1}}.  \label{msg-siz-pf-inequality-3}
	%	\end{equation}
	Thus, for any $w\in[0,\frac{1}{N^{K-1}}]$, we have 
	\begin{equation}
	\frac{1}{N^2-2(N-1)}>w. \label{msg-siz-pf-inequality-4}
	\end{equation}
	From \eqref{msg-siz-pf-inequality-2} and \eqref{msg-siz-pf-inequality-4}, we obtain \eqref{msg-siz-pf-inequality-1}, and thus the argument \eqref{msg-siz-pf-inequality-0} is true. 
	Then \eqref{msg-siz-pf-leakage} and \eqref{msg-siz-pf-inequality-0} imply that 
	\begin{equation}
	I(W_k;A_{1:N}^{[k']}Q_{1:N}^{[k']})>wL,
	\end{equation}
	which contradicts the individual leakage constraint definition in \eqref{def-weak}. 
	Thus, the supposition at beginning that $L\leq (N-2)\frac{L}{\log_{|\cX|}|\cY|}$ is false, and this proves the lemma. 
\end{proof}
%
%Similarly, we can verify P1 and P2 for the individual leakage case. 
%\begin{enumerate}[i)]
%	\item Since \eqref{converse-weak-proof-1} is the same as \eqref{converse-download}, this is the same as i) for the total leakage case. 
%	
%	\item When \eqref{converse-w-iteration-0} is equality, we also obtain for $k=1$ that 
%	\begin{equation}
%	H(A_{1:N}^{[1]}|W_{1}A_n^{[1]}\bF)=0,
%	\end{equation}
%	which is exactly the condition that proves P2 in \cite{Tian-Sun-Chen-PIR-IT19}. 
%\end{enumerate}
%Now we have proved the lower bound of message size which is $L\geq (N-1)\log|\cY|$. 

%%==================================================================
\section{Proof of \Cref{lemma-strong-iteration}}\label{lemma-strong-iteration-proof}
To prove the lemma, we consider the following. 
\begin{align}
&NH(A_{1:N}^{[k]},W_{1:k}|\bF)\geq \sum_{n=1}^{N}H(A_n^{[k]},W_{1:k}|\bF)  \label{converse-s-iteration-1} \\
&=\sum_{n=1}^{N}H(A_n^{[k]},W_{1:k}|Q_n^{[k]}) \label{converse-s-iteration-2} \\
&=\sum_{n=1}^{N}H(A_n^{[k+1]},W_{1:k}|Q_n^{[k+1]})  \label{converse-s-iteration-3} \\
&=\sum_{n=1}^{N}H(A_n^{[k+1]},W_{1:k}|\bF)  \label{converse-s-iteration-4} \\
&=\sum_{n=1}^{N}H(A_n^{[k+1]}|W_{1:k},\bF)+NH(W_{1:k})   \\
&\geq \sum_{n=1}^{N}H(A_n^{[k+1]}|W_{1:k},\bF A_{1:n-1}^{[k+1]})+NH(W_{1:k})  \label{converse-s-iteration-5} \\
&=H(A_{1:N}^{[k+1]}|W_{1:k},\bF)+NH(W_{1:k})   \\
&= H(A_{1:N}^{[k+1]},W_{1:k}|\bF)+(N-1)H(W_{1:k})  \\
&= H(A_{1:N}^{[k+1]},Q_{1:N}^{[k+1]},W_{1:k+1}|\bF)+(N-1)H(W_{1:k})  \\
&= H(A_{1:N}^{[k+1]},W_{1:k+1}|\bF)+k(N-1)L, 
\end{align}
where \eqref{converse-s-iteration-2} and \eqref{converse-s-iteration-4} follow from the Markov chain $(A_n^{[k]},W_{[1:K]})\rightarrow Q_n^{[k]}\rightarrow\bF$, 
and \eqref{converse-s-iteration-3} follows from identical distribution constraint in \eqref{privacy-identical-distribute}. 
The lemma is proved.

\bibliographystyle{ieeetr}

\end{document}